\documentclass[12pt]{article}

\pdfoutput=1

\makeatletter
\newif\ifkp@upRm% is used in the .fd-file of jkp
\DeclareSymbolFont{Letters}{OML}{jkp}{m}{n}
\DeclareMathSymbol{\partialup}{\mathord}{Letters}{128}
\DeclareMathSymbol{\DD}{D}{Letters}{128}
\makeatother

\usepackage{amsmath}
\usepackage{amssymb}
\usepackage{graphicx}
\usepackage[dvipsnames]{xcolor}
\usepackage{soul}
\numberwithin{equation}{section}
\usepackage{array}
\usepackage{dsfont}

\usepackage{tikz}\usetikzlibrary{matrix,fit}
\usepackage{tikz-cd} 
\usepackage{varwidth}
\usepackage{enumerate}
\usepackage{appendix}
\usepackage{xfrac}
\usepackage{nicefrac}
\usepackage{mathtools,slashed}

\usepackage{geometry}
\usepackage{nicematrix}

\usepackage[
    backend=bibtex,
    style=alphabetic,
maxbibnames=99,
giveninits=true,
minalphanames=1,
maxalphanames=3,url=false
  ]{biblatex}

\bibliography{Tdual.bib}

\usepackage{mathabx}
\usepackage{empheq}

\usepackage{setspace}

\usepackage{slashed}
\usepackage{upgreek}

\usepackage{tabstackengine}

\usepackage{wrapfig}
\usepackage[abs]{overpic}

\usepackage{float}

\DeclareMathAlphabet{\dutchcal}{U}{dutchcal}{m}{n}
\SetMathAlphabet{\dutchcal}{bold}{U}{dutchcal}{b}{n}
\DeclareMathAlphabet{\dutchbcal}{U}{dutchcal}{b}{n}

\fixTABwidth{T}

\newdimen\mytextwidth
\newcommand\rem[2][cyan!40!green]{\noindent\nobreak\hfil\penalty1000\hfilneg
\mytextwidth=\linewidth\advance\mytextwidth by 2mm
\begin{tikzpicture}[baseline=-\the\dimexpr\fontdimen22\textfont2\relax]\node[outer sep=0pt,draw=black,fill=#1,fill opacity=1,text opacity=1,rectangle,rounded corners]{\begin{varwidth}{\mytextwidth}\textcolor{white}{#2}\end{varwidth}};
\end{tikzpicture}\allowbreak
}

\newcommand\whiterem[2][white!]{\noindent\nobreak\hfil\penalty1000\hfilneg
\mytextwidth=\linewidth\advance\mytextwidth by 2mm
\begin{tikzpicture}[baseline=-\the\dimexpr\fontdimen22\textfont2\relax]\node[outer sep=0pt,draw=black,fill=#1,fill opacity=1,text opacity=1,rectangle,rounded corners,line width=1.5pt]{\begin{varwidth}{\mytextwidth}\textcolor{black}{#2}\end{varwidth}};
\end{tikzpicture}\allowbreak
}

\makeatletter
\newsavebox{\@brx}
\newcommand{\llangle}[1][]{\savebox{\@brx}{\(\m@th{#1\langle}\)}%
  \mathopen{\copy\@brx\kern-0.5\wd\@brx\usebox{\@brx}}}
\newcommand{\rrangle}[1][]{\savebox{\@brx}{\(\m@th{#1\rangle}\)}%
  \mathclose{\copy\@brx\kern-0.5\wd\@brx\usebox{\@brx}}}
\makeatother

\newcommand{\dd}{\partialup}

\newcommand{\CP}{\mathbb{CP}}
\newcommand{\CC}{\mathbb{C}}

\newcommand{\ggoth}{\mathfrak{g}}

\newcommand{\hgoth}{\mathfrak{h}}

\newcommand{\bea}{\begin{equation}}
\newcommand{\eea}{\end{equation}}
\newcommand{\bear}{\begin{eqnarray}}
\newcommand{\eear}{\end{eqnarray}}
\newcommand{\bearr}{\begin{eqnarray*}}
\newcommand{\eearr}{\end{eqnarray*}}

\usepackage{ytableau}
\ytableausetup{centertableaux}

\usepackage{bm}

\usepackage{tikz}
\usepackage{xparse}
\NewDocumentCommand{\xrightarrows}{ O{}O{} }{%
\mathrel{%
\vcenter{\hbox{%
\begin{tikzpicture}
  \node[minimum width=1cm,minimum height=1ex,anchor=south,align=center] (a){\text{\vphantom{hg}#1}\\[0.5ex] \vphantom{hg}#2};
  \draw[<-] ([yshift=0.35ex]a.west) -- ([yshift=0.35ex]a.east);
  \draw[->] ([yshift=-0.35ex]a.west) -- ([yshift=-0.35ex]a.east);
\end{tikzpicture}
}}%
}%
}

\renewbibmacro{in:}{}

\usepackage{mdframed}

\newbibmacro{string+doi}[1]{%
    \iffieldundef{doi}
      {%
        \iffieldundef{url}
          {#1}
          {\href{\thefield{url}}{#1}}%
      }
      {\href{https://doi.org/\thefield{doi}}{#1}}%
  }

\DeclareFieldFormat{title}{\usebibmacro{string+doi}{\mkbibemph{#1}}}
\DeclareFieldFormat[article]{title}{\usebibmacro{string+doi}{\mkbibquote{#1}}}

\newmdenv[
  topline=false,
  bottomline=false,
  rightline=false,
  linewidth=2pt,
  skipabove=\topsep,
  skipbelow=\topsep
]{siderules}

\newmdenv[
  topline=false,
  bottomline=false,
  linewidth=2pt,
  skipabove=\topsep,
  skipbelow=\topsep
]{siderulesright}

\makeatletter
\renewcommand{\@seccntformat}[1]{\csname the#1\endcsname.\quad}
\makeatother

\usepackage{xpatch}

\makeatletter
\renewcommand{\@chap@pppage}{
  \clear@ppage
  \thispagestyle{plain}
  \if@twocolumn\onecolumn\@tempswatrue\else\@tempswafalse\fi
  \null\vfil
  \markboth{}{}
  {\centering
   \interlinepenalty \@M
   \normalfont
   \MakeUppercase \appendixpagename\par}
  \if@dotoc@pp
    \addappheadtotoc
  \fi
  \vfil\newpage
  \if@twoside
    \if@openright
      \null
      \thispagestyle{empty}
      \newpage
    \fi
  \fi
  \if@tempswa
    \twocolumn
  \fi
}
\makeatother

\definecolor{navycol}{RGB}{100,150,160}
   \definecolor{pinkcol}{RGB}{242,55,55}
   \definecolor{greencol}{RGB}{50,205,50}

   \definecolor{bluecol}{RGB}{30,144,255}

 \allowdisplaybreaks

\usepackage{titlesec}

\titleformat*{\section}{\large\bfseries}
\titleformat*{\subsection}{\normalsize\bfseries}
\titleformat*{\subsubsection}{\normalsize\bfseries}
\titleformat*{\paragraph}{\large\bfseries}
\titleformat*{\subparagraph}{\large\bfseries}
\titlespacing{\author}{-5pt}{-5pt}{-5pt}[-5pt]

\makeatletter
\renewcommand\subsubsection{\@startsection{subsubsection}{3}{\z@}
                                     {-3.25ex\@plus -1ex \@minus -.2ex}
                                     {-1.5ex \@plus -.2ex}
                                     {\normalfont\normalsize\bfseries}}
\renewcommand\subsection{\@startsection{subsection}{3}{\z@}
                                     {-3.25ex\@plus -1ex \@minus -.2ex}
                                     {-1.5ex \@plus -.2ex}
                                     {\normalfont\normalsize\bfseries}}                                     
\makeatother

\usepackage{pgfplots}
\pgfplotsset{compat=1.15}
\usepackage{mathrsfs}
\usetikzlibrary{arrows}
\usepackage{afterpage}
\usepackage{emptypage}

\DeclareFontFamily{U}{solomos}{}
\DeclareErrorFont{U}{solomos}{m}{n}{10}
\DeclareFontShape{U}{solomos}{m}{n}{
  <-> s*[1.1]  gsolomos8r
}{}

   \usepackage{stmaryrd}

\usepackage{tikz}
\usetikzlibrary{arrows.meta}

\usepackage{indentfirst}

\usepackage{tocloft}
\let \savenumberline \numberline
\def \numberline#1{\savenumberline{#1.}}

\usepackage{etoolbox}
\patchcmd{\tableofcontents}{\@starttoc}{\vspace{-0.3cm}\@starttoc}{}{}

\usepackage{accents}
\newcommand\thickbar[1]{\accentset{\rule{.7em}{.7pt}}{#1}}
\newcommand\smallthickbar[1]{\accentset{\rule{.5em}{.7pt}}{#1}}
\usepackage{hyperref}
\hypersetup{
colorlinks=true,
linkcolor=MidnightBlue,
citecolor=violet,
filecolor=purple,
urlcolor=cyan,
breaklinks=true
}

\renewcommand{\bar}{\thickbar}

\usepackage{upgreek}

\usepackage{amsthm}

\newcounter{Chapcounter}

\newcommand{\chapter}[1] 
{ {\centering          
  \addtocounter{Chapcounter}{1} \Large \underline{\sffamily \texorpdfstring{\textbf{  Chapter \theChapcounter: ~#1}}{Lg}} }   
  \addcontentsline{toc}{section}{ \color{MidnightBlue} \texorpdfstring{Chapter ~}{Lg}\theChapcounter.\texorpdfstring{~~}{Lg} #1 }    
}

\usepackage{scalerel}

\newtheorem{prop}{Proposition}
\newcommand{\RNum}[1]{\uppercase\expandafter{\romannumeral #1\relax}}

\newtheorem{corollary}{Corollary}

\newtheorem{definition}{Definition}

\title{\textbf{Anomalous symmetries in K\"ahler geometry} \vspace{0.7cm}}

\author{Dmitri Bykov$^{\,a,\,b,\,c,\,d}$\footnote{Emails:
 bykov@mi-ras.ru, dmitri.v.bykov@gmail.com} \qquad\qquad Andrew Kuzovchikov$^{\,a,\,b,\,c}$\footnote{Email:
 andrkuzovchikov@mail.ru}
\\  \vspace{0cm}  \\
{\small $a)$ 
\emph{Steklov
Mathematical Institute of Russian Academy of Sciences,}} \\{\small \emph{Gubkina str. 8, 119991 Moscow, Russia} }\\
{\small $b)$ 
\emph{Institute for Theoretical and Mathematical Physics,}} \\{\small \emph{Lomonosov Moscow State University, 119991 Moscow, Russia}}\\
{\small $c)$ \emph{HSE University, 6 Usacheva str., Moscow 119048, Russia}}\\
{\small $d)$ \emph{Beijing Institute of Mathematical Sciences and Applications (BIMSA),}} \\{\small \emph{Huairou District, Beijing
101408, China}}
}

\date{}

\begin{document}

\begin{titlepage}
\maketitle
\thispagestyle{empty}

\begin{abstract}
We initiate a study of centrally extended (\emph{anomalous}) symmetries in K\"ahler geometry, focusing on the simplest, and most ubiquitous, Abelian case. In particular, we provide a local description of the geometry admitting such isometries. A long-standing no-go theorem asserts that there is an obstruction to gauging such symmetries in the purely K\"ahler framework. Utilizing the language of sypersymmetry, we then show that these symmetries may be gauged within the setup of generalized K\"ahler geometry.  Our results may be applied to quotients, T-dualities, etc.
\end{abstract}

\end{titlepage}

\section{Introduction}

In the present paper we study 2D sigma models with $\mathcal{N}=(2,2)$ supersymmetry (SUSY). This amount of SUSY imposes stringent constraints on the target spaces of these sigma models.  Historically, the first, and arguably most important, examples of such models involved K\"ahler target spaces~\cite{Zumino}. From the modern standpoint these are sigma models with this amount of SUSY and a vanishing torsion\footnote{I.e., a vanishing or topological $B$-field.}. 
However, non-vanishing torsion is also possible: this corresponds to bi-Hermitian geometry~\cite{Gates}, 
which was later rediscovered and reformulated in mathematical terms as generalized K\"ahler geometry~\cite{HitchinGeneralized, Gualtieri, GualtieriKahler}. 

Just like in any field theory, an important aspect of sigma models is the introduction of gauge fields for the relevant symmetries. From a mathematical perspective, this is tightly related to the theory of generalized K\"ahler quotients (cf.~\cite{Hitchin_1987}). Even in the purely K\"ahler setting~\cite{BaggerWitten}, which will mostly concern us in the present paper, this is a rather non-trivial question. It was realized early in~\cite{HullRocek} that not every algebra of K\"ahler isometries can be gauged in superspace: there is an obstruction, which, as we shall elaborate, is related to the symmetry algebra being \emph{centrally extended}. It is natural to call such symmetries \emph{anomalous}, drawing a parallel with the more well-known anomalies tightly related to central extensions. For example, the Adler-Bell-Jackiw anomaly has been famously related to the central extension of the Lie algebra of gauge transformations~\cite{Faddeev, FaddeevShatashvili}, the central extension of the Virasoro algebra is governed by the Weyl anomaly in 2D conformal models (cf.~\cite{GSW}), etc.\footnote{See~\cite{Gritskov:2026yyi} for a recent discussion of perturbative anomalies in quantum mechanics and \cite{Gamayun:2026zav} for a cohomological computation of the conformal anomaly in $\mathrm{J}\widebar{\mathrm{J}}$-deformed 2D conformal models.} Unlike these anomalies, which are inherently quantum, the central extensions that we will describe are classical. In this regard, our setup is somewhat reminiscent of the results of~\cite{BrownHenneaux}, where central extensions in asymptotically $AdS_3$ gravity were considered, or of those in~\cite{OliveWitten, LosevShifman}, where classical central extensions of supersymmetry algebras were studied\footnote{Central extensions have also been studied in classical mechanics: some examples include the Galilei algebra~\cite{Bargmann, Leblond}, the symmetry algebra of a particle in a magnetic field (cf.~\cite{BykovKrivorol}), etc.}.

All of the symmetry algebras that we will encounter are finite-dimensional. There are various restrictions on the Lie algebras that may appear in the finite-dimensional setting:  in particular, Whitehead's lemma states that semi-simple algebras cannot have a non-trivial central extension (cf.~\cite{GuilleminSternberg1990symplectic, Landshoff}). The finite-dimensional Lie algebras that \emph{can} be centrally extended are the solvable ones. Out of these, Abelian ones are the most widely encountered. In this paper we will therefore restrict to Abelian algebras of K\"ahler isometries. There are plenty of examples of K\"ahler manifolds with Abelian isometries: such are, for example, the complex tori, as well as other toric manifolds. An important example, which served as an initial stimulant for this work, is the so-called $\eta$-deformed $\CP^n$ geometry, cf.~\cite{DelducMagroVicedo, Litvinov2019, Demulder, BykovLust} for more on this. 

It was found in~\cite{BKK} that, using a certain trick, in many cases the obstruction to gauging mentioned above may be overcome. An important observation was that the presence of a central extension is intimately linked to  the non-invariance of the K\"ahler potential under the symmetry transformations. It was also pointed out that, in several relevant cases, the potential could be made invariant by introducing auxiliary twisted chiral fields\footnote{Twisted chiral and semi-chiral superfields were introduced in~\cite{Gates} and \cite{BuscherRocek} respectively. The relation of various types of multiplets to the complex structures of generalized K\"ahler geometry has been elaborated in~\cite{LindstromOffShellComplete}.} entering via generalized K\"ahler transformations (i.e., via total derivatives). This turns out to be sufficient to gauge these symmetries in the generalized K\"ahler setting,  using a certain novel gauge multiplet\footnote{Another way of gauging in this setting makes use of the so-called Large Vector Multiplet of~\cite{LindstromVectorMultiplets, LindstromGeneralizedNonabelian}. The relation between the two types of multiplets has been elaborated in~\cite{BykovLindstromRocek}.} $(\mathbf{V}, \mathbf{X})$ consisting of a real gauge superfield $\mathbf{V}$ and a semi-chiral one $\mathbf{X}$.

The goal of this paper is to extend these results to the general case. To this end, we first characterize the most general K\"ahler geometry admitting an Abelian algebra of holomorphic isometries that is subject to a central extension. It turns out that it is always a K\"ahler quotient of a certain canonical local geometry, which we call $\mathcal{N}$, w.r.t. a (non-centrally extended) Abelian subalgebra. Since the latter  quotient may be performed using standard techniques, it suffices to study the gauging of the canonical geometry~$\mathcal{N}$. As we will show, the (centrally extended) symmetry  algebra of $\mathcal{N}$ is really a sum of several copies of the Heisenberg algebra (each being a central extension of $\mathbb{R}^2$). It follows that one only needs to specify the gauging of the Heisenberg algebra, which is what we ultimately do.

\vspace{0.3cm}
The paper is organized as follows. We start in Section~\ref{Kahcentralext} by recalling that the Lie algebra of moment maps for the action of a Lie group on a K\"ahler manifold may admit a central extension. For the case of an Abelian Lie algebra we prove an important relation between the central extension and the (non)invariance of the K\"ahler potential. In Section~\ref{metricssec} we describe a general local theory for metrics admitting central extensions, showing that they arise as K\"ahler quotients of certain very explicit model metrics. We then proceed in Section~\ref{gaugingsec} to discuss how such anomalous symmetries can be gauged. To this end one passes to the generalized K\"ahler geometry setup by introducing auxiliary twisted chiral fields that `cancel the anomaly' and subsequently gauges the relevant symmetries. Applications to (generalized) quotients as well as T-dualities are discussed.

\section{Central extensions in  K\"ahler geometry} \label{Kahcentralext} Consider a general K\"ahler manifold with complex coordinates $\mathrm{z}_i$ ($i=1 \cdots N$), K\"ahler potential $\mathcal{K}$ and the K\"ahler form
\begin{align}
    \omega={i\over 2}\sum\limits_{j, k}^N\,{\dd^2\mathcal{K}\over \dd \mathrm{z}_j \dd \smallthickbar{\mathrm{z}}_k}\,d\mathrm{z}_j \wedge d\smallthickbar{\mathrm{z}}_k\,.
\end{align}
Suppose we have a holomorphic Killing vector field 
\begin{align}
  \mathsf{V}=\mathsf{V}^i(\mathrm{z}) \frac{\dd}{\dd \mathrm{z}_i}+\widebar{\mathsf{V}}^i(\smallthickbar{\mathrm{z}})\frac{\dd}{\dd \smallthickbar{\mathrm{z}}_i}=\mathsf{V}^{(1,0)} + \mathsf{V}^{(0,1)}\,,  
\end{align}
where $\mathsf{V}^{(1,0)}$ and $\mathsf{V}^{(0,1)}$ are the $(1,0)$ and $(0,1)$ parts of $\mathsf{V}$, respectively. Notice that the coefficients $\mathsf{V}^i(\mathrm{z})$ are holomorphic as a consequence of the fact that $\mathsf{V}$ preserves the complex structure, $\pounds_{\mathsf{V}} \mathcal{J}=0$. 

In this case the condition that $\mathsf{V}$ be Killing, i.e. $\pounds_{\mathsf{V}} \mathcal{G}=0$, is equivalent to the fact that it be a symplectomorphism,  $\pounds_{\mathsf{V}} \omega=0$, since the two structures are related by $\mathcal{G}(\bullet, \bullet')=\omega(\bullet, \mathcal{J} \bullet')$. In other words\footnote{Here and below by $\nabla_\mathsf{V}$ we mean the action of the vector field on a function: ${\nabla_\mathsf{V} \mathcal{K}\equiv {(\mathsf{V}^j \dd_{\mathrm{z}_j}+\widebar{\mathsf{V}}^j \dd_{\smallthickbar{\mathrm{z}}_j}) \mathcal{K}}}\equiv \pounds_{\mathsf{V}} \mathcal{K}$. Sometimes we will also use the notation $\nabla_\mathsf{V}^{(1,0)} \mathcal{K}=\mathsf{V}^j \dd_{\mathrm{z}_j}\mathcal{K}$ and analogously for the $(0,1)$ part. Since $\mathsf{V}^j(z)$ and $\mathsf{W}^j(z)$  are holomorphic, one has $\bigl[\nabla_\mathsf{V}^{(1,0)}, \nabla_\mathsf{W}^{(0,1)}\bigr]=0$.}, ${\pounds_{\mathsf{V}} \omega=d(i_\mathsf{V} \omega)={i\over 2} \dd \bar{\dd} (\nabla_\mathsf{V} \mathcal{K})=0}$, so that
\begin{align}\label{Kvar}
   \nabla_\mathsf{V} \mathcal{K}=f_\mathsf{V}(\mathrm{z})+\widebar{f}_\mathsf{V}(\smallthickbar{\mathrm{z}})\,, 
\end{align}
where $f_\mathsf{V}(\mathrm{z})$ is holomorphic. We see that, despite the fact that the metric is invariant under~$\mathsf{V}$, the K\"ahler potential does not have to be invariant but may instead be subjected to a K\"ahler transformation. This has an interesting effect on the Lie algebra of moment maps $\upmu_\mathsf{V}$, which are the Hamiltonians for the action of the vector fields, defined via  $i_\mathsf{V} \omega=d \upmu_\mathsf{V}$. In case the manifold is not simply-connected, $i_\mathsf{V} \omega$ might not be exact so that the moment map may exist only locally\footnote{As an example, consider an $\mathrm{S}^1$-action on a 2-torus $\mathbb{T}^2$. In some situations one can overcome the problem by defining a moment map with values in $\mathrm{S}^1$ \cite{McDuff_1988,Ortega_2004}.}. As we will explain shortly, local considerations will suffice for our purposes.  The explicit expression for the moment map is 
\begin{align}\label{momap}
    \upmu_\mathsf{V}:={i\over 2}\left(\mathsf{V}^j\dd_{\mathrm{z}_j} \mathcal{K}-\widebar{\mathsf{V}}^j \dd_{\smallthickbar{\mathrm{z}}_j} \mathcal{K}+\widebar{f}_\mathsf{V}(\smallthickbar{\mathrm{z}})-f_\mathsf{V}(\mathrm{z})\right)\,.
\end{align}
%It is a real valued function due to~(\ref{Kvar}). 
By a direct calculation, one finds a complete description of the Lie algebra of moment maps w.r.t. the Poisson bracket:

\begin{prop} [\cite{Souriau, Libermann_1987, HullRocek, Arnold_1978}]
    Let $\upmu_\mathsf{V}, \upmu_\mathsf{W}$ be the moment maps corresponding to the holomorphic Killing vector fields $\mathsf{V}, \mathsf{W}$. Their Poisson bracket is\footnote{The functions $f_\mathsf{V}$ and $f_\mathsf{W}$ are holomorphic, thus, $\nabla_\mathsf{V} f_\mathsf{W}=\nabla_\mathsf{V}^{(1,0)}f_\mathsf{W}$ and  $\nabla_{\mathsf{W}} f_\mathsf{V}=\nabla_\mathsf{W}^{(1,0)}f_\mathsf{V}$.} 
    \begin{align} \label{PoissonBracket}
    &\Bigl\{ \upmu_\mathsf{W}, \upmu_\mathsf{V} \Bigr\}=\upmu_{[\mathsf{W}, \mathsf{V}]}+\dutchcal{C}(\mathsf{W}, \mathsf{V})\,,\\ \label{CentralExtension} & \textrm{where}\quad \dutchcal{C}(\mathsf{W}, \mathsf{V})=i \left(\nabla_{\mathsf{V}} f_\mathsf{W}-\nabla_{\mathsf{W}} f_\mathsf{V}+f_{[\mathsf{W}, \mathsf{V}]}\right)\,.
    \end{align}
    Moreover, $\dutchcal{C}(\mathsf{W}, \mathsf{V})$ is a (real) constant. 
\end{prop}

The fact that $\dutchcal{C}$ is real follows from the reality of moment maps. Therefore, if $\ggoth$ denotes the (real) Lie algebra of holomorphic vector fields,   one gets a central extension $\dutchcal{C}\in \mathfrak{g}^\ast\wedge \mathfrak{g}^\ast$. The extended Lie algebra is labeled $\widehat{\ggoth}=\ggoth\,\oplus\, \mathbb{R}\,\dutchcal{C}$, where by $\oplus$ one means a sum of vector spaces. A central extension is called trivial if there exists a $d\in \ggoth^\ast$ such that $\dutchcal{C}(\mathsf{W}, \mathsf{V})=d([\mathsf{W}, \mathsf{V}])$. In this case a mere shift of all moment maps, $\mu_{\mathsf{V}}\mapsto \mu_{\mathsf{V}}-d(\mathsf{V})$, eliminates the central extension in~(\ref{PoissonBracket}).

The expression for $\dutchcal{C}$ is invariant under the transformation $f_\mathsf{V}\mapsto f_\mathsf{V}+\nabla_\mathsf{V} h(\mathrm{z})$, which in turn is generated by a K\"ahler transformation $\mathcal{K} \mapsto \mathcal{K}+h(\mathrm{z})+\widebar{h}(\smallthickbar{\mathrm{z}})$. Since the K\"ahler potentials in two different patches of a given K\"ahler manifold $\mathcal{M}$ are related by a K\"ahler transformation, the absence or presence of a central extension is a local condition and may be studied in a single patch. If $\mathrm{dim}_{\CC}\, \mathcal{M}=N$, we may thus assume that we are given a set of vector fields in a domain $\mathrm{U}\subset \CC^N$.

\subsection{Cohomological interpretation.} The set of functions $f_{\mathsf{V}}$ ($\mathsf{V}\in \ggoth$) in~(\ref{Kvar}) may be regarded as a map 
\begin{align}
    f:\quad \mathfrak{g}\mapsto \mathcal{R}:=\mathcal{O}(\mathrm{U})\,,
\end{align}
where $\mathcal{O}(\mathrm{U})$ is the sheaf of holomorphic functions in a domain $\mathrm{U}\subset\CC^N$. Moreover, $\mathcal{R}$ may be thought of as a representation of $\mathfrak{g}$, where the Lie algebra acts by vector fields. In these terms $f$ is a 1-cochain on $\mathfrak{g}$ with values in $\mathcal{R}$. The central extension may then be written as 
\begin{align}\label{deltadef}
    \dutchcal{C}=\delta f\,,\quad\quad \textrm{where}\quad\quad \delta: \mathfrak{g}^\ast\otimes \mathcal{R}\mapsto (\mathfrak{g}^\ast)^{\wedge 2 } \otimes \mathcal{R}\,
\end{align}
is the Chevalley-Eilenberg differential 
defined in~(\ref{CentralExtension}).

Another important point has to do with the fact that $\dutchcal{C}(\mathsf{V}, \mathsf{W})\in \mathbb{R}$ is a real constant: taking into account that $\mathbb{R}\subset \mathcal{R}$ (i.e., a trivial representation embedded in $\mathcal{R}$ by constant functions), we may refine~(\ref{deltadef}) by stating that
\begin{align}
    \dutchcal{C}=\delta f \in (\mathfrak{g}^\ast)^{\wedge 2 } \otimes \mathbb{R}\,.
\end{align}
Thus,   whenever $\dutchcal{C}\nequiv 0$, one has a cocycle $\bigl[\dutchcal{C} \bigr]\in H^2(\ggoth, \mathbb{R})$,  since $\delta \dutchcal{C}=\delta^2 f=0$. It is trivial in the $(\ggoth, \mathbb{R})$-complex only if $f\in \mathfrak{g}^\ast\otimes \mathbb{R}$, i.e. if $f_\mathsf{V}$ is a real constant for any $\mathsf{V}$.  It then follows from the definition of $\dutchcal{C}$ that in this case $\dutchcal{C}(\mathsf{W}, \mathsf{V})=i\,f_{[\mathsf{W}, \mathsf{V}]}$. Defining $d\in\ggoth^\ast$ by $d(\mathsf{V})=i\,f_{\mathsf{V}}$, one sees that this corresponds to a trivial central extension, as  defined above. 

\subsection{Vanishing central extension.}\label{Vanishing central extension}

A trivial central extension implies $\dutchcal{C}=\delta f\equiv 0$,  i.e. that~$f$ is a  cocycle, $f\in H^1(\ggoth, \mathcal{R})$. A trivial cocycle is $f=\delta h$, which is the same as $f_\mathsf{V} = \nabla_\mathsf{V} h(\mathrm{z})$ for all $\mathsf{V}\in \ggoth$. In this case, to produce an invariant K\"ahler potential, one simply shifts $\mathcal{K} \mapsto \mathcal{K}-h(\mathrm{z})-\widebar{h}(\smallthickbar{\mathrm{z}})$. 

From now on let us assume that the Lie algebra~$\mathfrak{g}$ is Abelian.  We will now prove that, in K\"ahler geometry, the vanishing of the central extension for Abelian\footnote{Note that, for non-Abelian $\ggoth$, in general this is not true.} $\mathfrak{g}$ implies that the cocycle $f$ is trivial, i.e. $\bigl[f \bigr]=0$. We will need the following notion of linear independence of vector fields:

\begin{definition}
    A finite family of vector fields $\bigl\{v_A\bigr\}_{A=1}^m$ in an open domain $\mathrm{U}$ is called linearly independent over $\mathbb{C}$ (or $\mathbb{R}$) iff for every point $p \in \mathrm{U}$ the vectors $\bigl\{v_A\big|_p\bigr\}_{A=1}^m$ are linearly independent over $\mathbb{C}$ (or $\mathbb{R}$).
\end{definition}

\vspace{0.1cm}
Here $v_A$ could either stand for the real vector fields $\mathsf{V}_A$, in which case one considers linear (in)dependence over $\mathbb{R}$, or their $(1,0)$ parts $\mathsf{V}_A^{(1,0)}$, where one is interested in linear (in)dependence over $\mathbb{C}$.

\vspace{0.2cm}
Let us start with the following proposition:

\begin{prop}\label{linindepvectorfieldsprop}
Assume the holomorphic parts of Killing vector fields $\mathsf{V}_A$ are linearly independent over $\mathbb{C}$. Then the cocycle $f$ is trivial, $\bigl[ f \bigr]=0$.
\end{prop}
\begin{proof}
Under the assumption of linear independence in an open domain $\mathrm{U}$, the vector fields may be rectified to the canonical form  $\mathsf{V}_A={\dd_{\mathrm{z}_A}}, A=1, \ldots, m$ \cite[Chapter~9]{Lee_2012}. The vanishing of the central extension takes the form 
\begin{align}
   \dutchcal{C}(\mathsf{V}_A, \mathsf{V}_B)={\dd \over \dd \mathrm{z}_A} f_{\mathsf{V}_B} -{\dd \over \dd \mathrm{z}_B} f_{\mathsf{V}_A}=0\,, \quad\quad \textrm{so that}\quad\quad  f_{\mathsf{V}_A}={\dd \over \dd \mathrm{z}_A} h
\end{align}
by the Poincar\'e lemma. 
\end{proof}

Now we are ready to prove the statement announced above: 
\begin{prop}\label{abeliantrivcocycleprop}
    Assume the Lie algebra $\ggoth=\mathrm{Span}(\mathsf{V}_1, \ldots, \mathsf{V}_n)$ (with $\mathsf{V}_A$ linearly independent over $\mathbb{R}$) is Abelian. Additionally, suppose $\dutchcal{C}\equiv 0$, so that 
one has  
a 1-cocycle $f\in H^1(\ggoth, \mathcal{O}(\mathrm{U}))$. Then  
$\mathsf{V}_1^{(1,0)}, \ldots, \mathsf{V}_n^{(1, 0)}$ are linearly independent and therefore $\bigl[ f \bigr]=0$.
\end{prop}

\begin{proof}
Suppose that the rank $\mathrm{rk}(\mathsf{V}_A^i)=m\leq n$. Assume additionally that the first $m$ rows are linearly independent (or else we permute the labels). 
It then follows from Proposition~\ref{linindepvectorfieldsprop} that one can make the K\"ahler potential invariant w.r.t. the vector fields $\mathsf{V}_1, \ldots, \mathsf{V}_m$. In this case the Hermitian matrix
\begin{align}
    g_{A\smallthickbar{B}}:= \nabla_{\mathsf{V}_A}^{(1,0)}\nabla_{\mathsf{V}_{B}}^{(0,1)} \mathcal{K}\,,\quad\quad A, B=1, \ldots, m
\end{align}
is, in fact, symmetric (and real). Indeed, due to the invariance of the  potential ${\nabla_{\mathsf{V}_{B}}^{(0,1)} \mathcal{K}=-\nabla_{\mathsf{V}_{B}}^{(1,0)} \mathcal{K}}$, which leads to   $g_{A\smallthickbar{B}}=-\nabla_{\mathsf{V}_{B}}^{(1,0)}\nabla_{\mathsf{V}_{A}}^{(1,0)} \mathcal{K}$ due to the commutativity of the vector fields, and finally $g_{A\smallthickbar{B}}=g_{B\smallthickbar{A}}$ again due to the invariance of the potential w.r.t. $\mathsf{V}_{A}$.

Now, the $(1,0)$-parts of the remaining vector fields take the form
\begin{align}
   \mathsf{V}_{m+k}^{(1,0)}=\sum_{B=1}^m\,\alpha^B_{k} 
   \,\mathsf{V}_{B}^{(1,0)}\,,\quad\quad k=1, \ldots n-m\,,
\end{align}
where the coefficients $\alpha^B_{m+k}$ are holomorphic.  Computing $c(\mathsf{V}_A, \mathsf{V}_{m+k})$ from the definition (see (\ref{CentralExtension})), one finds
\begin{align}\label{vancmixed}
    &0=i\,c(\mathsf{V}_A, \mathsf{V}_{m+k})=\nabla_{\mathsf{V}_A}^{(1,0)} \nabla_{\mathsf{V}_{m+k}} \mathcal{K} - \nabla_{\mathsf{V}_{m+k}}^{(1,0)} \nabla_{\mathsf{V}_A} \mathcal{K}=\\ \nonumber
    &=\nabla_{\mathsf{V}_A}^{(1,0)}\nabla_{\mathsf{V}_{m+k}}^{(0,1)} \mathcal{K} - \nabla_{\mathsf{V}_{m+k}}^{(1,0)} \nabla_{\mathsf{V}_A}^{(0,1)} \mathcal{K}=\sum_{B=1}^m\,g_{A\smallthickbar{B}}\,\bigl(\thickbar{\alpha}^{\smallthickbar{B}}_{k}-\alpha^B_{k}\bigr)\,.
\end{align}
In passing to the second line, we  used the commutativity of the vector fields, whereas in writing the last equality, we used the symmetry of $g_{A\smallthickbar{B}}$.

The fact that the metric is positive-definite implies that $g_{A\smallthickbar{B}}$ is positive-definite as well. It then follows from~(\ref{vancmixed}) that 
$\alpha_{k}^B-\widebar{\alpha}_{k}^{\smallthickbar{B}}\equiv 0$. By holomorphicity, the latter  implies that $\alpha_{k}^B$ are real constants, so that the Killing vector fields $\mathsf{V}_{m+k}$ are linear  combinations of the $\mathsf{V}_A$'s. As such, they may be dropped. 
It follows that the cocycle is trivial by Proposition~\ref{linindepvectorfieldsprop}. 
\end{proof}

The following corollary will be important when we discuss the gauging of  centrally extended isometry algebras:

\begin{corollary}
For Abelian $\mathfrak{g}$ the  K\"ahler potential can be made invariant (by a K\"ahler transformation) if and only if $\dutchcal{C}\equiv 0$.
\end{corollary}
In the next section, we provide some examples where the central extension cannot be removed.

\subsection{Flat space examples.}
Let us start from a general flat Hermitian metric in~$\CC^N$, with its corresponding K\"ahler form and K\"ahler potential:
\begin{align}
ds^2=\sum\limits_{j,k=1}^N\,a_{j\widebar{k}}d\mathrm{z}_j d\smallthickbar{\mathrm{z}}_k\,,\quad \omega={i\over 2}\sum\limits_{j,k=1}^N\,a_{j\widebar{k}}d\mathrm{z}_j \wedge d\smallthickbar{\mathrm{z}}_k\,,\quad \mathcal{K}=\sum\limits_{j,k=1}^N\,a_{j\widebar{k}}\mathrm{z}_j\smallthickbar{\mathrm{z}}_k\,,
\end{align}
where $||a_{j\widebar{k}}||_{j,\widebar{k}}$ is a Hermitian positive-definite matrix. Now let us consider the isometry group generated by (holomorphic) shift isometries 
\begin{align}
    \mathrm{z}_j\mapsto \mathrm{z}_j+n_j\,\epsilon\,,\quad\quad n\in \CC^N\quad \textrm{and}\quad \epsilon \in \mathbb{R}\,.
\end{align}
The corresponding real vector field and the moment map for this action have the form
\begin{align}
    \mathsf{V}^{(n)}=n_j {\dd\over \dd \mathrm{z}_j}+\smallthickbar{n}_j{\dd\over \dd \smallthickbar{\mathrm{z}}_j}\,,\quad\quad\quad \upmu_n=i \sum\limits_{j, k}a_{j\widebar{k}}\Bigl(n_j \smallthickbar{\mathrm{z}}_k-\mathrm{z}_j \smallthickbar{n}_k\Bigr)\,.
\end{align}
The Poisson bracket of moment maps corresponding to the $n-$ and $m-$ shifts is
\begin{align}
    \bigl\{\upmu_n, \upmu_m \bigr\}=\dutchcal{C}(\mathsf{V}^{(n)}, \mathsf{V}^{(m)})=i \bigl((m,n)-(n,m)\bigr)\,,\quad\quad \textrm{where} \quad\quad(m, n)=a_{j\widebar{k}}n_j \smallthickbar{m}_k\,.
\end{align} 

\vspace{0.3cm}
The two main examples where the central extension is non-vanishing are as follows:

\vspace{0.3cm} \noindent
\textit{1)} $\mathcal{M}=\CC$, 
where $n=1, m=i$ and $\bigl\{\upmu_n, \upmu_n \bigr\}= 2 a_{1\widebar{1}}$. In this case the fact that the central extension must be present follows from the form of the vector fields, i.e. $\mathsf{V}_1=\mathrm{Re}\left( {\dd_{\mathrm{z}_1}}\right)$ and $\mathsf{V}_2=\mathrm{Re}\bigl(i {\dd_{\mathrm{z}_1}}\bigr)$. Indeed, $\mathsf{V}_1^{(1,0)}$ and $\mathsf{V}_2^{(1,0)}$ are  linearly dependent, so that, by Proposition~\ref{abeliantrivcocycleprop}, \textit{any} K\"ahler metric invariant under these vector fields will inevitably lead to a central extension.

\vspace{0.3cm} \noindent
\textit{2)} 
$\mathcal{M}=\CC^2$, where $n=\begin{pmatrix}
        i & 0
    \end{pmatrix}, \, m=\begin{pmatrix}
        0 & i
    \end{pmatrix}$ and $\Bigl\{\upmu_n, \upmu_m \Bigr\}=-2\, \mathrm{Im}(a_{1\widebar{2}})$. Unlike the example of point \textit{1)}, here the presence of a central extension is not dictated by the form of the vector fields but depends on the choice of the K\"ahler metric.

\section{Metrics admitting a central extension}\label{metricssec}

Let us denote by $\bigl\{\mathsf{V}_A \bigr\}_{A=1}^n$ the holomorphic Killing vector fields generating the Abelian Lie algebra~$\ggoth$. 
Then one has
\begin{align}
    &\Bigl[\mathsf{V}_A,\mathsf{V}_B \Bigr] = 0\,,&& \nabla_{\mathsf{V}_A}f_{\mathsf{V}_B} - \nabla_{\mathsf{V}_B} f_{\mathsf{V}_A} = i \,\dutchcal{C}_{AB} \equiv i\,\dutchcal{C}(\mathsf{V}_A,\mathsf{V}_B)\,, \label{abelianAlgebra}
\end{align}
One can perform a linear transformation on the vector fields, using a matrix ${g\in \mathrm{GL}(n,\mathbb{R})}$, i.e. $\mathsf{V}_B \mapsto \mathsf{V}'_A = g_{AB}\mathsf{V}_B$. The function $f_\mathsf{V}$ is linear in $\mathsf{V}$, therefore, the matrix $\dutchcal{C}$ of $\dutchcal{C}_{AB}$'s transforms under the map as $\dutchcal{C}\mapsto \dutchcal{C}' = g\,\dutchcal{C}\,g^t$. Using this transformation, one brings the anti-symmetric matrix $\dutchcal{C}$ to the canonical form (here $\mathds{1}_n$ is an $n\times n$ identity matrix)
\begin{align}
\dutchcal{C}_{k, n} = c\,
    \begin{pmatrix}
        \mathds{1}_k \otimes \varepsilon & \\
        & 0\cdot \mathds{1}_{n-2k}
    \end{pmatrix}\,,\quad\quad \textrm{where}\quad\quad  \varepsilon = \begin{pmatrix}
    0 & 1\\
    -1 & 0
\end{pmatrix}\,.\label{standardFormOfC}
\end{align}
To be able to keep track of the terms sourcing the central extension we prefer to keep the parameter $c$ in $\dutchcal{C}$.

\subsection{The Heisenberg algebra.}\label{HeisenbergAlgebra}
It is clear from (\ref{standardFormOfC}) that the most basic example is given by the centrally extended two-dimensional abelian algebra, i.e. the Heisenberg algebra~$\mathfrak{h}_3$: this corresponds to the case  $n=2$ and $k=1$. 

Let us call the respective vector fields  $\mathsf{V}$ and $\mathsf{W}$. There are two possible local models for these vector fields.   
The \textit{first} case is when $\mathsf{V}^{(1,0)}$ and $\mathsf{W}^{(1,0)}$ are linearly independent: then they can be brought to the form
\begin{align}\label{linindepvecfields}
   (1)\quad\quad  \mathsf{V}^{(1,0)} = {\dd \over \dd \mathrm{z}}\,,\quad\quad \mathsf{W}^{(1,0)} = {\dd \over \dd \mathrm{w}}
\end{align}
using some holomorphic coordinates $\mathrm{z}$ and $\mathrm{w}$.
The \textit{second} case is when they are linearly dependent, so that 
\begin{align}\label{lindepvecfields}
   \hspace{1cm} (2)\quad\quad  \mathsf{V}^{(1,0)} = {\dd \over \dd \mathrm{z}}\,,\quad\quad \mathsf{W}^{(1,0)} = \uplambda(\mathrm{r})\,{\dd \over \dd \mathrm{z}}\,,
\end{align}
where $\uplambda(\mathrm{r})$ is a holomorphic function of the remaining coordinates\footnote{We assume that $\uplambda(r)$ is not a constant real number, since otherwise $\mathsf{W}$ is proportional to $\mathsf{V}$.}. 
One can now construct a local model of the K\"ahler potential in each of the two cases:
\begin{prop}\label{HeisenbergProp}
    Let $\mathsf{V}$ and $\mathsf{W}$ be commuting Killing vector fields, $\bigl[\mathsf{V}, \mathsf{W} \bigr]=0$, and the corresponding moment maps $\upmu_\mathsf{V}, \upmu_\mathsf{W}$ satisfy $\bigl\{\upmu_\mathsf{V}, \upmu_\mathsf{W} \bigr\}=\dutchcal{C}\left(\mathsf{V}, \mathsf{W}\right)\equiv c$. Then the K\"ahler potential may be brought to one of the following forms:
    \begin{align}\label{abelianKah1}
        &\displaystyle (1)\quad \mathcal{K}=-{ic\over 2}(\mathrm{w}+\smallthickbar{\mathrm{w}})(\mathrm{z}-\smallthickbar{\mathrm{z}})+\mathcal{K}_0(\mathrm{w}-\smallthickbar{\mathrm{w}}, \mathrm{z}-\smallthickbar{\mathrm{z}}, \cdots)\quad\quad \textrm{or}\\ \label{abelianKah2}
        &\displaystyle (2)\quad\mathcal{K}={ic\over 2}{\bigl(i(\mathrm{z}-\smallthickbar{\mathrm{z}})\bigr)^2\over \uplambda(\mathrm{r})-\smallthickbar{\uplambda}(\smallthickbar{\mathrm{r}})}+
    \mathcal{K}_0(\mathrm{r}, \smallthickbar{\mathrm{r}})\,,
    \end{align}
    depending on whether $\mathsf{V}^{1,0}, \mathsf{W}^{1,0}$ are linearly independent (case (1)) or linearly dependent (case (2)).
\end{prop}
\begin{proof}
As before, we will set $f_\mathsf{V}=0$, which means that  $\mathcal{K}=\mathcal{K}(\mathrm{z}-\smallthickbar{\mathrm{z}}, \ldots)$. The central extension is then $\dutchcal{C}(\mathsf{V}, \mathsf{W})=i \dd_{\mathrm{z}} f_\mathsf{W}=c \in \mathbb{R}$. This integrates to
\begin{align}
   f_\mathsf{W}=-ic\, \mathrm{z}+f_0(\,\cdots)\,, 
\end{align}
where $\cdots$ denotes all coordinates other than $z$. Thus, ${\nabla_\mathsf{W} \mathcal{K}=-ic\,(\mathrm{z}-\smallthickbar{\mathrm{z}})+f_0+\widebar{f_0}}$. 

We will now separately consider the two cases above. In the first case, assuming $f_0=f_0(\mathrm{w}, \mathrm{r})$, the last two terms in $\nabla_\mathsf{W} \mathcal{K}$ can be compensated by a K\"ahler transformation with the function $\kappa(\mathrm{w},\mathrm{r})=\int^\mathrm{w}\,d\mathrm{x} \,f_0(\mathrm{x}, \mathrm{r})$, so that we are left with the equation 
\begin{align}
    \left({\dd \over \dd \mathrm{w}}+{\dd \over \dd \smallthickbar{\mathrm{w}}}\right) \mathcal{K}=-ic\,(\mathrm{z}-\smallthickbar{\mathrm{z}})\,,
\end{align}
whose solution is~(\ref{abelianKah1}).  
In the second case we have 
\begin{align}
    &\nabla_\mathsf{W} \mathcal{K}=(\uplambda - \smallthickbar{\uplambda})\, \mathcal{K}'(\mathrm{z}-\smallthickbar{\mathrm{z}}, \cdots)=-ic\,(\mathrm{z}-\smallthickbar{\mathrm{z}})+(f_0+\widebar{f_0})\,,\quad\quad \textrm{so that}\\ \label{Kahpottoruslinterm}
    &\displaystyle\mathcal{K}=-{ic\over 2(\uplambda-\smallthickbar{\uplambda})}(\mathrm{z}-\smallthickbar{\mathrm{z}})^2+(\mathrm{z}-\smallthickbar{\mathrm{z}})\frac{f_0+\widebar{f_0}}{\uplambda-\smallthickbar{\uplambda}}+\mathcal{K}_0(\cdots)\,,
\end{align}
where $f_0$ and $\mathcal{K}_0$ are functions of the remaining coordinates. Shifting $\mathrm{z}\mapsto \mathrm{z}+{1\over i c} f_0$ one brings this to the canonical form\footnote{For $c=0$ the non-invariant piece of the potential is an example of a non-trivial cohomology class $\bigl[ f \bigr]\neq 0$. However, as was shown in the proof of Proposition~\ref{abeliantrivcocycleprop}, it leads to a non-positive metric.}~(\ref{abelianKah2}), 
    which has the structure of a fibration. 
\end{proof} 

\subsubsection{Examples.} Let us consider some regular geometries illustrating case $(2)$. First, if $\uplambda\equiv \mathrm{const.}$, one has the metric product geometry $\CC_z\times \cdots$. Assuming periodic identifications of $z$, one may replace $\CC_z$ with $\CC_z^\ast$ or $\mathbb{T}^2_z$.

Another interesting case is $\uplambda(\mathrm{r})= \mathrm{r}$. Here the K\"ahler potential~(\ref{abelianKah2}), and hence the resulting metric, has an apparent singularity at $\mathrm{Im}(\mathrm{r})=0$. In order for the metric to be complete, this singularity has to be at infinite distance. If $K_0$ is regular at $\mathrm{Im}(r)=0$, this does not hold at $\mathrm{Im}(z)=0$. Thus, we are led to consider $K_0$ that is itself singular at $\mathrm{Im}(r)=0$. The mildest singularity that resolves the issue is such that $K_0\sim \log{(\mathrm{Im}(r))}$, which gives the potential
\begin{align}\label{KahlerBerndtPotential}
    \mathcal{K}=c\,\frac{(i(\mathrm{z}-\widebar{\mathrm{z}}))^2}{i(\mathrm{r}-\widebar{\mathrm{r}})}+a\,\log{(i(\mathrm{r}-\widebar{\mathrm{r}}))}\,,\quad\quad a=\mathrm{const.}>0
\end{align}

This geometry was discovered in\footnote{A reprint of the latter paper can be found in~\cite{KahlerWorks}. See also~\cite{Molitor} for a historical review.}~\cite{Berndt, KahlerIndividuum}  in the study of automorphic forms and later reappeared in the context of supergravity in~\cite{Louis}. 

As it stands, the potential~(\ref{KahlerBerndtPotential}) gives a metric on the total space of a $\CC$-bundle, with fiber coordinate $\mathrm{z}$, over the hyperbolic plane $\mathbb{H}_2$, with coordinate $\mathrm{r}$. There are several obvious isometries:
\begin{align}
    &\mathrm{z}\mapsto \mathrm{z}+\epsilon_1+\mathrm{r}\,\epsilon_2\,,\quad \mathrm{r}\mapsto \mathrm{r}+\epsilon_3\,,\quad \mathrm{r}\mapsto s^2\,\mathrm{r}\,, \mathrm{z}\mapsto s\,\mathrm{z}\,,\quad \\ \nonumber  &\quad\quad\quad\quad\quad\quad\quad\quad\quad\quad\quad \textrm{where}\quad \epsilon_{1,2,3}\in \mathbb{R}\,, s\in \mathbb{R}_+
\end{align}
Remarkably, the metric is also invariant under the inversion $\mathrm{r}\mapsto -{1\over \mathrm{r}}, \mathrm{z}\mapsto {\mathrm{z}\over \mathrm{r}}$ (the shift in the K\"ahler potential being $\delta \mathcal{K}=i\,c\,\bigl({\smallthickbar{\mathrm{z}}^2\over \smallthickbar{\mathrm{r}}}-{\mathrm{z}^2\over \mathrm{r}}\bigr)-a\,\log{(\mathrm{r}\smallthickbar{\mathrm{r}})}$). Together with the above symmetries, this implies invariance w.r.t. the full group $\mathrm{SL}(2, \mathbb{R})\ltimes \CC$. As a result, the underlying manifold is a homogeneous space
\begin{align}
    \mathcal{M}_{\mathrm{SJ}}:=\frac{\mathrm{SL}(2, \mathbb{R})\ltimes \CC}{\mathrm{SO}(2)}
\end{align}
The central extension in question affects the $\CC$-factor: since the Lie algebra of Killing vector fields is $\ggoth=\mathfrak{sl}_2(\mathbb{R})\ltimes \mathbb{R}^2$, its centrally extended counterpart is ${\widehat{\ggoth}=\mathfrak{sl}_2(\mathbb{R})\ltimes \hgoth_3}$, where $\hgoth_3$ is the Heisenberg algebra.

Let us also mention that one could  convert the complex plane $\CC$ into a 2-torus with modular parameter $\mathrm{r}$ by imposing the identifications
\begin{align}
    z\sim z+m +n\,\mathrm{r}\,,\quad\quad m,n\in \mathbb{Z}\,.
\end{align}
The quotient manifold $\mathcal{M}_{\mathrm{SJ}}\bigl/\mathbb{Z}^2$ is known as the universal elliptic curve. It is the fiber bundle over $\mathbb{H}_2$, whose fiber at a given point $\mathrm{r}\in \mathbb{H}_2$ is the elliptic curve with modulus~$\mathrm{r}$.

\subsection{The central extension on quotient manifolds.}
 
It turns out that case $(2)$ of Proposition~\ref{HeisenbergProp} can be reduced to a special subcase of case $(1)$. Indeed, suppose one has the linearly dependent vector fields~(\ref{lindepvecfields}) and a K\"ahler potential $\mathcal{K}(\mathrm{z}, \smallthickbar{\mathrm{z}}, \mathrm{r}, \smallthickbar{\mathrm{r}})$ for which they are Killing. Introducing an auxiliary complex coordinate $\mathrm{u}$, one constructs the extended potential\footnote{We have chosen to add a quadratic function in $\mathrm{u}-\smallthickbar{\mathrm{u}}$ for the sake of simplicity. Adding an arbitrary convex function would lead to the same conclusions.}
\begin{align}\label{primepot}
    \mathcal{K}':=\mathcal{K}+{1\over 2}\Bigl(i(\mathrm{u}-\smallthickbar{\mathrm{u}})\Bigr)^2
\end{align}
as well as the extended vector fields
\begin{align}\label{vectfieldextension}
    \widehat{\mathsf{V}}^{(1,0)} = {\dd \over \dd \mathrm{z}}\,,\quad\quad \widehat{\mathsf{W}}^{(1,0)} = {\dd \over \dd \mathrm{u}}+\uplambda(\mathrm{r})\,{\dd \over \dd \mathrm{z}}\,,
\end{align}
which are Killing for $\mathcal{K}'$. Moreover, $\mathcal{K}'$  admits an additional Killing vector field 
\begin{align}
    {\widehat{\mathsf{U}}^{(1,0)}={\dd \over \dd \mathrm{u}}}=\widehat{\mathsf{W}}^{(1,0)}-\uplambda(\mathrm{r})\, \widehat{\mathsf{V}}^{(1,0)}\,.
\end{align}
It is also easy to characterize the central extension for~$\mathcal{K}'$: one has 
\begin{align}
    \dutchcal{C}(\widehat{\mathsf{V}}, \widehat{\mathsf{W}})=c\,,\quad\quad \dutchcal{C}(\widehat{\mathsf{V}}, \widehat{\mathsf{U}})=\dutchcal{C}(\widehat{\mathsf{W}}, \widehat{\mathsf{U}})=0\,.
\end{align}

It is equally simple to invert this construction: to get back to $\mathcal{K}$ one should take a quotient of $\mathcal{K}'$ w.r.t. the vector field $\widehat{\mathsf{U}}$. An explicit construction of the quotient is elaborated in Appendix~\ref{appKahquot}. 

We will also need the following generalization of this result:
\begin{prop}\label{cquotprop}
    Let $\mathsf{U}$ be a Killing vector field on $\mathcal{M}$ and $\ggoth$ the Lie algebra of remaining isometries, $[\ggoth, \mathsf{U}]=0$. Moreover, assume that $\mathsf{U}\in \mathrm{Ker}\,\dutchcal{C}$, that is $\dutchcal{C}(\bullet, \mathsf{U})=0$. Consider the K\"ahler quotient $ \mathcal{M}\bigl/\!\!\bigl/\mathbb{R}$ by the group $\mathbb{R}$ generated by $\mathsf{U}$, and denote by $\pi_\ast(\ggoth)$ the projection of~$\ggoth$ to the quotient. In this case $\dutchcal{C}\bigl|_{\pi_\ast(\ggoth)}=\dutchcal{C}\bigl|_{\ggoth}$. 
\end{prop}
\begin{proof}
    At first, let us rectify $\mathsf{U}$, i.e. $\mathsf{U}=\dd_{\mathrm{u}}+\dd_{\smallthickbar{\mathrm{u}}}$. As $\mathsf{U}\in \mathrm{Ker}\,c$, we may assume that the K\"ahler potential $\mathcal{K}$ is invariant under the action of $\mathsf{U}$. Explicitly, ${\mathcal{K} = \mathcal{K}(\mathrm{u}-\smallthickbar{\mathrm{u}},\mathrm{z},\smallthickbar{\mathrm{z}})}$, where $\mathrm{z}$ represents coordinates other than $\mathrm{u}$. We will denote the generators of~$\ggoth$ 
    by~$\bigl\{\mathsf{V}_A\bigr\}_{A=1}^m$. Then, there exists a function $\mathrm{P}(\mathrm{z},\smallthickbar{\mathrm{z}})$ such that $\mathrm{x} = \mathrm{u}-\smallthickbar{\mathrm{u}} - i\,\mathrm{P}(\mathrm{z},\smallthickbar{\mathrm{z}}) 
    $ satisfies the following constraints:
    \begin{align}
        &\nabla_{\mathsf{V}_A}\, \mathrm{x} = 0\,,&& \nabla_{\mathsf{U}}\, \mathrm{x} = 0\,. \label{SystemDiff}
    \end{align}
    
    Indeed, by virtue of the Frobenius theorem \cite{Olver_1995} locally there exist ${n=\mathrm{dim}_\mathbb{R} \mathcal{M}-m-1}$ solutions of (\ref{SystemDiff}), say $\bigl\{F_a(\mathrm{u}-\smallthickbar{\mathrm{u}},\mathrm{z},\smallthickbar{\mathrm{z}})\bigr\}_{a=1}^{n}$, such that their differentials are linearly independent\footnote{A set of differentials $\{dF_a\}_{a=1}^n$ is called linearly independent in an open domain $\mathrm{U}$ iff $d F_1\wedge dF_2\wedge\dots \wedge d F_n \neq 0$ for all points of $\mathrm{U}$.}. In writing $F_a$ as a function of $\mathrm{u}-\smallthickbar{\mathrm{u}}$ we explicitly solved the equation $\nabla_{\mathsf{U}}\,F_a = 0$. Every function of the $F_a$'s is a solution to (\ref{SystemDiff}) as well. Then, the desired function $\mathrm{P}(\mathrm{z},\smallthickbar{\mathrm{z}})$ exists as a consequence of the implicit function theorem iff at least for one of the $F_a$'s its derivative $\dd_\mathrm{y} F_a \neq 0$, where $\mathrm{y}:=\mathrm{u}-\smallthickbar{\mathrm{u}}$. This is the case due to the linear independence of the differentials of the $F_a$'s\footnote{$F_a$'s may be thought of as coordinates on the space of orbits of the isometry group. If $\dd_{\mathrm{y}}F_a = 0$ for all $a$, one of the isometries in $\ggoth$ is generated by $\mathrm{Re}\left(\uplambda(\mathrm{z})\dd_{\mathrm{u}}\right)$ but then by Proposition \ref{linindepvectorfieldsprop} it has a non-trivial central extension with $\mathsf{U}$.}.

    Now, with a given $\mathrm{x}$ we will treat the K\"ahler potential as a function of $\mathrm{x}$ and $\mathrm{z}$'s, i.e. $\mathcal{K}(\mathrm{u}-\smallthickbar{\mathrm{u}},\mathrm{z},\smallthickbar{\mathrm{z}}) =\widehat{\mathcal{K}}\left(\mathrm{x},\mathrm{z},\smallthickbar{\mathrm{z}}\right)$. We will also decompose $\mathsf{V}_A$ into two parts:
    \begin{align}
        \mathsf{V}_A = \widehat{\mathsf{V}}_A + \sigma_A (\mathrm{z}) \frac{\dd}{\dd \mathrm{u} } + \widebar{\sigma}_A (\smallthickbar{\mathrm{z}}) \frac{\dd}{\dd \smallthickbar{\mathrm{u}} }\,,
    \end{align}
    where $\widehat{\mathsf{V}}_A$ does not contain $\dd_\mathrm{u}$ or $\dd_{\smallthickbar{\mathrm{u}}}$ components. Note that $\widehat{\mathsf{V}}_A = \pi_*\left(\mathsf{V}_A\right)$ and 
    \begin{align}
        \label{nablaK}
        \nabla_{\mathsf{V}_A}\mathcal{K} = \nabla_{\widehat{\mathsf{V}}_A} \widehat{\mathcal{K}}\big|_\mathrm{x} = f_{\mathsf{V}_A} + \widebar{f}_{\mathsf{V}_A}
    \end{align}
    because $\nabla_{\mathsf{V}_A} \,\mathrm{x} =0$.

    Now let us turn to the K\"ahler quotient. In order to find the K\"ahler potential on the quotient manifold, one replaces $\mathrm{x}$ with $i\mathbf{V}$ and performs a Legendre transform w.r.t.~$\mathbf{V}$~\cite{Hitchin_1987, LindstromReview}. Thus one should express  $i\mathbf{V}$ as a function of $\mathrm{z},\smallthickbar{\mathrm{z}}$ from the equation
    \begin{align}
        &\frac{\dd}{\dd \mathbf{V}}\widehat{\mathcal{K}}\bigl(i\mathbf{V},\mathrm{z},\smallthickbar{\mathrm{z}}\bigr)+\xi = 0\,,
        \label{equationOnV}
    \end{align}
    where $\xi \in \mathbb{R}$ is the Fayet-Iliopoulos (FI) parameter. 
    The resulting K\"ahler potential on the quotient manifold is
    \begin{align}
        \mathcal{K}' = \widehat{\mathcal{K}}\bigl(i\mathbf{V}(\mathrm{z},\smallthickbar{\mathrm{z}}),\mathrm{z},\smallthickbar{\mathrm{z}}\bigr)+\xi\, \Bigl(\mathbf{V}(\mathrm{z},\smallthickbar{\mathrm{z}})+\mathrm{P}(\mathrm{z},\smallthickbar{\mathrm{z}})\Bigr)\,.
    \end{align}
    Let us now compute the action of $\widehat{\mathsf{V}}_A$ on $\mathcal{K}'$:
    \begin{align}
        \nabla_{\widehat{\mathsf{V}}_A} \mathcal{K}' = f_{\mathsf{V}_A} - i\xi \sigma_A + \widebar{f}_{\mathsf{V}_A} + i\xi\widebar{\sigma}_A + \left(\frac{\dd}{\dd \mathsf{V}} \widehat{\mathcal{K}} + \xi\right)\cdot \nabla_{\widehat{\mathsf{V}}_A}\mathbf{V}(\mathrm{z},\smallthickbar{\mathrm{z}}) =\\
        = f_{\mathsf{V}_A} - i\xi \sigma_A + \widebar{f}_{\mathsf{V}_A} + i\xi\widebar{\sigma}_A\,.\nonumber
    \end{align}
    Here we have used~(\ref{SystemDiff}), (\ref{nablaK}), (\ref{equationOnV}) and the fact that $\dd_\mathrm{u} f_{\mathsf{V}_A} = 0$, which follows from ${\dutchcal{C}\left(\mathsf{V}_A, \mathsf{U}\right) = 0}$. 
    As a result of a direct calculation, one finds
    \begin{align}
        \dutchcal{C}(\mathsf{V}_A,\mathsf{V}_B) - \dutchcal{C}(\widehat{\mathsf{V}}_A, \widehat{\mathsf{V}}_B) = \xi \left(\nabla_{\widehat{\mathsf{V}}_A} \sigma_B (\mathrm{z}) - \nabla_{\widehat{\mathsf{V}}_B} \sigma_A (\mathrm{z})\right) = 0\,,
    \end{align}
    where the last equality is a consequence of  $\bigl[\mathsf{V}_A,\mathsf{V}_B \bigr] = 0$. Therefore, $\dutchcal{C}\bigl|_{\pi_*(\ggoth)}=\dutchcal{C}\bigl|_{\ggoth}$.
\end{proof}

\subsection{General case.}\label{KahlerpotentialGeneralCaseSect}

It is now clear how to treat the general case when the central extension is given by the matrix~(\ref{standardFormOfC}). First, assume that the first $2k$ vector fields, corresponding to the centrally extended block, are such that their $(1,0)$-components are linearly independent. Let us group these fields as 
\begin{align}\label{VWfields}
   &{\mathsf{V}^{(1,0)}_1=\dd_{\mathrm{z}_1}, \;\mathsf{W}^{(1,0)}_1=\dd_{\mathrm{w}_1}, \;\mathsf{V}^{(1,0)}_2=\dd_{\mathrm{z}_2}, \;\mathsf{W}^{(1,0)}_2=\dd_{\mathrm{w}_2}, \;\;\ldots}\quad \\&\textrm{so that}\quad  \dutchcal{C}(\mathsf{V}_A, \mathsf{W}_B)=c\,\delta_{AB}, \quad A, B=1, \ldots k \label{VWcentext}
\end{align}
On top of $\mathsf{V}_A$ and $\mathsf{W}_B$ one has the additional $n-2k$ Killing vector fields $\mathsf{U}_1, \ldots, \mathsf{U}_{n-2k}$, such that $\mathsf{U}_{\bullet}\in \mathrm{Ker}\,\dutchcal{C}$. 
The fact that $\dutchcal{C}(\mathsf{U}_\ell, \mathsf{V}_B)=0$ and $\dutchcal{C}(\mathsf{V}_A, \mathsf{V}_B)=0$ implies that one can make the potential invariant w.r.t. the $\mathsf{U}_\ell$ and $\mathsf{V}_A$ fields simultaneously. In other words,
\begin{align}
    \mathcal{K}=\mathcal{K}\Bigl(\mathrm{z}_1-\smallthickbar{\mathrm{z}}_1, \cdots ; \mathrm{w}, \smallthickbar{\mathrm{w}}, \mathrm{r}, \smallthickbar{\mathrm{r}}\Bigr)\,,
\end{align}
where $\mathrm{r}, \smallthickbar{\mathrm{r}}$ denote all the remaining coordinates. Additionally, due to $\mathsf{U}_\ell$-invariance, 
\begin{align}
    \nabla_{\mathsf{U}_\ell} \mathcal{K}=0\,,\quad\quad \ell=1, \ldots, n-2k
\end{align}
which further constrains the function $\mathcal{K}$. Then, in order to incorporate the central extension~(\ref{VWcentext}), one follows the proof of Proposition~\ref{HeisenbergProp}, and a simple extrapolation of~(\ref{abelianKah1}) gives the local form of the K\"ahler potential:
\begin{align}\label{anomalousKgen}
   \mathcal{K}=-{ic\over 2}\,\sum_{A=1}^k\,\Bigl(\mathrm{w}_A+\smallthickbar{\mathrm{w}}_A \Bigr)\Bigl(\mathrm{z}_A-\smallthickbar{\mathrm{z}}_A \Bigr)+\mathcal{K}_0\Bigl(\mathrm{w}_1-\smallthickbar{\mathrm{w}}_1, \mathrm{z}_1-\smallthickbar{\mathrm{z}}_1, \cdots ; \mathrm{r}, \smallthickbar{\mathrm{r}}\Bigr)+\\  \nonumber +\kappa(\mathrm{w}, \mathrm{r})
    +\smallthickbar{\kappa}(\smallthickbar{\mathrm{w}}, \smallthickbar{\mathrm{r}})\,,
\end{align}
where the function $\kappa$ represents  a K\"ahler transformation featuring in the proof of Proposition~\ref{HeisenbergProp}. As we shall see shortly, this geometry plays a distinguished role, so we introduce a special notation for it:

\begin{definition}
    Let the triple 
\begin{align}
   (\mathcal{N}, \ggoth, \dutchcal{C}_{k, n}) 
\end{align}
denote the K\"ahler manifold $\mathcal{N}$ with K\"ahler potential~(\ref{anomalousKgen}) and an action of an Abelian Lie algebra $\ggoth$ (s.t.  $\mathrm{dim}_{\mathbb{R}}\,\ggoth=n$), whose central extension is given by $\dutchcal{C}_{k, n}$ (cf.~(\ref{standardFormOfC})). 
\end{definition}
The manifold $\mathcal{N}$ can have an arbitrary complex dimension $D\geq 2k$, depending on the number of $\mathrm{r}$ variables, which we do not specify here. In other words, $\mathcal{N}$ is a collective notation for a family of manifolds parametrized by $\mathcal{K}_0$ in~(\ref{anomalousKgen}).

\vspace{0.3cm}
The above analysis fully characterizes the local geometry in the case that the first~$2k$ vector fields are linearly independent. Thus, we are left to consider the situation when the first $2k$ vector fields are not necessarily linearly independent.  Generalizing the construction from the previous section, we arrive at the following:
\begin{prop}\label{uniformprop}
    Let $\mathcal{M}$ 
    be a K\"ahler manifold  with an isometric holomorphic action of an Abelian Lie algebra $\ggoth$ ($\mathrm{dim}_{\mathbb{R}} \,\ggoth=n$) with central extension $\dutchcal{C}_{k, n}$~[defined in~(\ref{standardFormOfC})]. Then, locally, $\mathcal{M}$ may be described as a K\"ahler quotient,
\begin{align}
    \mathcal{M}=\Bigl(\mathcal{N}, \;\ggoth\oplus \mathbb{R}^m, \;\dutchcal{C}_{k,\, n+m}\Bigr) \;\Bigl/\!\!\!\Bigl/\; \mathbb{R}^m\,,
\end{align}
    of $\mathcal{N}$ with K\"ahler potential~(\ref{anomalousKgen}) by the Killing vector fields $\Bigl\{\mathsf{U}'_1 \ldots \mathsf{U}'_m \Bigr\}\in \mathrm{Ker}\,\dutchcal{C}_{k, \,n+m}$. Moreover, any such quotient leads to a manifold $\mathcal{M}$ with isometry algebra $\mathfrak{g}$ and central extension $\dutchcal{C}_{k, n}$.
\end{prop}
\begin{proof}
The last statement follows by iteratively applying Proposition~\ref{cquotprop} $m$ times.  It thus remains to prove the first part, namely to construct the manifold $\mathcal{N}$, given $\mathcal{M}$ with a fixed central extension.

  Let us denote by $\bigl\{\mathcal{V}_a\bigr\}_{a=1}^{2k}$ the generators\footnote{For $\mathcal{M}=(\mathcal{N},\ggoth,\dutchcal{C}_{k,n})$, these are $\{\mathsf{V}_A, \mathsf{W}_A\}_{A=1}^k$ in the above notation.} of $\ggoth$ corresponding to the  non-trivial $\mathds{1}_k\otimes \varepsilon$ block of $\dutchcal{C}_{k,n}$ (cf.(\ref{standardFormOfC})). Their holomorphic parts are not necessarily linearly independent. We will write a matrix of components of the holomorphic parts: $\Gamma = ||\mathcal{V}_a^j||_{a,j}\,$, where ${a\in \bigl\{1,2,\cdots,2k\bigr\}}$ and $j\in \bigl\{1,2,\cdots,\dim_{\mathbb{C}}\mathcal{M}\bigr\}$. We will assume that the upper left $\mathrm{d}\times \mathrm{d}$-block of $\Gamma$ is non-degenerate, where $\mathrm{d} = \text{rank}\, \Gamma \leq 2k$. 

  Next, we will take $\mathcal{N} = \mathcal{M}\,\times\, \mathbb{C}^m$, where  $m=2k-\mathrm{d}$, and the holomorphic coordinates on the extra factor will be denoted $\bigl\{\mathrm{u}'_\ell\bigr\}_{\ell = 1}^m$. 
  Generalizing~(\ref{primepot}), we equip it with a K\"ahler metric arising from the potential
   \begin{align}\label{kahpotmod2}
        \widehat{\mathcal{K}} = \mathcal{K} + \sum_{\ell = 1}^m{1\over 2}\Bigl(i(\mathrm{u}'_\ell-\smallthickbar{\mathrm{u}}'_\ell)\Bigr)^2\,,
    \end{align}
    where $\mathcal{K}$ is the  K\"ahler potential of the original manifold $\mathcal{M}$. Mimicking~(\ref{vectfieldextension}), one also extends the vector fields $\bigl\{\mathcal{V}_a\bigr\}_{a=\mathrm{d}+1}^{2k}$ as follows: 
    \begin{align}
        \mathcal{V}_{\mathrm{d}+\ell} \mapsto \widehat{\mathcal{V}}_{\mathrm{d}+\ell} = \mathcal{V}_{\mathrm{d}+\ell} + {\dd \over \dd \mathrm{u}'_{\ell}} + {\dd \over \dd \smallthickbar{\mathrm{u}}'_{\ell}}\quad\;\text{for}\;\;\ell = 1,2,\cdots,m\,.
    \end{align}
    Together with all other generators in $\ggoth$ they form an algebra  $\ggoth'\simeq\ggoth$. By construction, $\widehat{\mathcal{K}}$ also admits the additional Killing vector fields
    \begin{align}
        \mathsf{U}'_\ell = \frac{\dd}{\dd \mathrm{u}'_\ell}+{\dd \over \dd \smallthickbar{\mathrm{u}}'_\ell}\,,\quad \text{for}\;\,\ell = 1,2,\cdots,m\,.
    \end{align}
    The algebra generated by $\ggoth'$ and $\bigl\{\mathsf{U}'_\ell\bigr\}_{\ell=1}^m$ is isomorphic to $\ggoth\,\oplus\,\mathbb{R}^m$. Computing the central extension, one verifies that $\dutchcal{C}|_{\ggoth'} = \dutchcal{C}|_{\ggoth}$ and $\mathsf{U}'_{\ell} \,\in\,\mathrm{Ker}\,c$ for all $\ell$. Therefore, the resulting central extension is of the $\dutchcal{C}_{k,n+m}$-type. Moreover, now the generators corresponding to the non-trivial block of $\dutchcal{C}_{k,n+m}$ have linearly independent holomorphic parts. Using the reasoning from the start of the present section, after a suitable change of coordinates and a K\"ahler transformation one can bring $\widehat{\mathcal{K}}$ to the form~(\ref{anomalousKgen}). 

    The original manifold $\mathcal{M}$ may be restored after performing a K\"ahler quotient with respect to $\bigl\{\mathsf{U}'_\ell\big\}_{\ell = 1}^m$. This simply amounts to dropping the $\mathrm{u}'_\ell$-dependent piece in~(\ref{kahpotmod2}). 
    Thus, we have constructed the desired triple ${\bigl(\mathcal{N}, \ggoth\,\oplus\, \mathbb{R}^m, \dutchcal{C}_{k,\, n+m}\bigr)}$.
\end{proof}

\section{Gauging with a central extension} \label{gaugingsec}

Why is it important to have an invariant K\"ahler potential? The answer comes from the realm of $\mathcal{N}=(2,2)$ supersymmetric sigma models\footnote{We refer to \cite{MirrorSymmetryBook} for details on this subject.}, which are intimately connected to K\"ahler geometry~\cite{Zumino}. 

First, we recall that $\mathcal{N}=(2,2)$ superspace involves the light-cone coordinates $x^\pm$ on a 2d worldsheet $\mathbb{R}^{1,1}$, as well as the Grassmann coordinates\footnote{$\theta^\pm$ transform as  Weyl fermions under Lorentz transformations.} $\theta^{\pm}, \thickbar{\theta^{\pm}}$. In order to define the various superfields, one introduces the superderivatives
\begin{align}
    &\mathcal{D}_{\pm}={\dd\over \dd \theta^{\pm}}-i\,\bar{\theta}^{\pm}\,{\dd \over \dd x^{\pm}}\,,&&\bar{\mathcal{D}}_{\pm}=-{\dd\over \dd \bar{\theta}^{\pm}}+i\,\theta^{\pm}\,{\dd \over \dd x^{\pm}}\,.
\end{align}
They form an algebra $\mathcal{D}_{\pm}^2=\bar{\mathcal{D}}_{\pm}^2=0\,,\;\{ \mathcal{D}_{\pm}, \bar{\mathcal{D}}_{\mp}\}=0\,,\;\{\mathcal{D}_{\pm}, \bar{\mathcal{D}}_{\pm}\}=2i \dd_{x^{\pm}}$. Next we define the superfields:
\begin{definition}
    Superfields are functions on superspace $(x^{\pm}, \theta^{\pm}, \thickbar{\theta^{\pm}})$ satisfying the following additional constraints:
\begin{align}\label{chir1}
&\textrm{Chiral:} &&\bar{\mathcal{D}}_+ \mathbf{Z}=\bar{\mathcal{D}}_- \mathbf{Z}=0\,,
    \\ \label{chir2}
    &\textrm{Twisted chiral:} &&\bar{\mathcal{D}}_+ \mathbf{S}=\mathcal{D}_- \mathbf{S}=0\,,
    \\ \label{chir3}
    &\textrm{Left semi-chiral:} &&\bar{\mathcal{D}}_+ \mathbf{X}=0\,,\\ \label{chir4}
    &\textrm{Right semi-chiral:} &&\mathcal{D}_- \mathbf{Y}=0\,.
\end{align}
\end{definition}

Now, given a K\"ahler manifold $\mathcal{M}$ and its K\"ahler potential $\mathcal{K}$, one considers the model defined by the following action functional:
\begin{align}
    \mathcal{S} = \int d x^+dx^-d \theta^+\,d\theta^-\,d\widebar{\theta}^+\,d\widebar{\theta}^- \;\mathcal{K}\left(\mathbf{Z},\widebar{\mathbf{Z}}\right)\,,\label{N=2,2Action}
\end{align}
where we have replaced holomorphic coordinates $\mathrm{z}$ on $\mathcal{M}$ with chiral superfields, which we will denote by  bold capital letters $\mathbf{Z}$. 

Let us stress an important point: the K\"ahler potential in (\ref{N=2,2Action}) is naturally defined up to an addition of total superderivative terms preserving reality conditions -- a generalized K\"ahler transformation -- because they yield zero upon integration. The K\"ahler transformation $\mathcal{K}\mapsto\mathcal{K}+f\left(\mathbf{Z}\right)+\widebar{f}\left(\widebar{\mathbf{Z}}\right)$ is an example of this. However, generalized K\"ahler transformations allow for greater freedom. For example, one can consider replacing $f$ by functions of chiral and twisted chiral fields:
\begin{align}\label{genKahtrans1}
    \mathcal{K}\mapsto\mathcal{K}+f(\mathbf{Z},\mathbf{S})+\widebar{f}\left(\widebar{\mathbf{Z}}, 
\widebar{\mathbf{S}}\right)+g(\mathbf{Z}, \widebar{\mathbf{S}})+\widebar{g}\left(\widebar{\mathbf{Z}}, 
\mathbf{S}\right)\,.
\end{align}
The additional terms are still total derivatives due to the constraints~(\ref{chir1})-(\ref{chir2}). This freedom will be essential for the constructions that follow.

\subsection{Gauging.}

Our main goal in this section will be to consider gauging of a group of holomorphic isometries of a K\"ahler manifold. Gauging is particularly important for the study of K\"ahler quotients or T-dualities  (see~\cite{AlvarezGaumeTduality, GiveonTduality} for a review, and~\cite{Cavalcanti} for the relation to generalized geometry).  In the former case, to find the K\"ahler potential for a quotient geometry one gauges the isometry group and eliminates the gauge fields using their equations of motion. In the latter case one also adds terms ensuring the gauge fields are pure gauge before eliminating them. We refer to \cite{Hitchin_1987, Ro_ek_1992} for a comprehensive review of these constructions.

If one wishes to consider gauging directly in $\mathcal{N}=(2,2)$ superspace, one can modify the differential constraints~(\ref{chir1})-(\ref{chir4}) by introducing appropriate covariant superderivatives featuring `superconnections'. In this case, provided the original K\"ahler potential is invariant w.r.t. global transformations, it will continue being invariant w.r.t. the gauge transformations, as it is a function of superfields and not of their derivatives. This is the main reason that we have investigated conditions for the invariance of the K\"ahler potential in great detail in Section~\ref{Kahcentralext}.

Even in the presence of a central extension, it is possible to make a K\"ahler potential invariant using additional compensator chiral fields. However, in the case of a non-trivial central extension, the gauging of such models is ill-defined \cite{HullRocek}. Therefore, one needs a different approach that was introduced in~\cite{BKK}. We will now explain that, in the general case of an Abelian isometry algebra $\mathfrak{g}$, one can make the K\"{a}hler potential invariant  using a generalized K\"ahler transformation involving twisted chiral fields\footnote{In a different context, a similar trick of performing a generalized K\"ahler transformation prior to gauging is also used in the recent paper~\cite{Bykov:2026ytj}.}.

\begin{prop}\label{twistedchiralprop}
    For a central extension $\dutchcal{C}_{k, n}$ of the form~(\ref{standardFormOfC}) the K\"ahler potential can be made invariant by introducing $k$ twisted chiral superfields and performing a generalized K\"ahler transformation (i.e., adding a total derivative).
\end{prop}

\begin{proof}
If $\dutchcal{C}\equiv  0$ (i.e. $k=0$), it follows from the results of~Section \ref{Vanishing central extension} that one can bring the K\"ahler potential to an invariant form.

When $\dutchcal{C}\nequiv 0$, let us use the same notation for vector fields as in Section \ref{KahlerpotentialGeneralCaseSect}, i.e.: one has the vector fields\footnote{Now the holomorphic parts of $\mathsf{V}$- and $\mathsf{W}$-fields  are not necessarily linearly independent.} $\bigl\{\mathsf{V}_A, \mathsf{W}_A\bigr\}_{A=1}^k$ plus $\bigl\{\mathsf{U}_i\bigr\}_{i=1}^{n-2k}$, such that $\dutchcal{C}(\mathsf{V}_A, \mathsf{W}_B) = \delta_{AB}$
whereas all other $\dutchcal{C}(\bullet,\bullet')$ are zero. Now we wish to find the shift $\mathcal{K}'= \mathcal{K}-h -\Bar{h}$ for some holomorphic function~$h$ that will eliminate as many $f_{\bullet}$'s as possible. To this end, one can choose, for example, $\bigl\{\mathsf{V}_{A} \bigr\}_{A=1}^{k}$ and $\bigl\{\mathsf{U}_i \bigr\}_{i = 1}^{n-2k}$. They form a maximal subalgebra 
\begin{align}
    \mathfrak{p}\subset \ggoth\quad\quad \textrm{such that}\quad\quad \dutchcal{C}|_{\mathfrak{p}} = 0\,,
\end{align}
i.e. a maximal isotropic subalgebra. Now we apply the technique from Section~\ref{Vanishing central extension} to~$\mathfrak{p}$. By Proposition~\ref{abeliantrivcocycleprop}, the cohomology class $\bigl[f\big|_{\mathfrak{p}}\bigr]=0$ is trivial. Thus, there exists a holomorphic function $h$ such that $f_\mathsf{V} = \nabla_{\mathsf{V}} h$ for $\mathsf{V}\in \mathfrak{p}$. As a result, we end up with the shifted K\"{a}hler potential $\mathcal{K}'$ 
that is invariant under the action of $\mathfrak{p}$, i.e., $\nabla_\mathsf{V} \mathcal{K}' = 0$ for all $\mathsf{V}\in \mathfrak{p}$. Note that the transformation law of the K\"{a}hler potential under the action of the other Killing vector fields is now modified.

Now, the (generalized) K\"{a}hler potential invariant under the full symmetry algebra is given by the following expression:
\begin{align}
    &\mathcal{K}'' = \frac{1}{2}\left(\prod_{A=1}^k\,e^{\mathbf{S}_{A} \nabla_{\mathsf{W}_{A}}}+\prod_{A=1}^k\,e^{\smallthickbar{\mathbf{S}}_{A} \nabla_{\mathsf{W}_{A}}}\right)\mathcal{K}'\,,\label{abelianShift}
\end{align}
The new potential is invariant w.r.t. to the simultaneous shifts of the original variables together with $\mathbf{S}$, i.e.
\begin{align}
    \left(\nabla_{\mathsf{W}_{A}}-{\dd\over \dd \mathbf{S}_{A}}-{\dd \over \dd\widebar{\mathbf{S}}_{A}}\right) \mathcal{K}''=0\,.\label{NewIsometry}
\end{align}
Note that we cannot subtract $f_{\mathsf{W}_{A}}$ from $\mathcal{K}'$ separately because $\nabla_{\mathsf{V}_{A}}f_{\mathsf{W}_{A}} \neq 0$ and it will inevitably violate $\nabla_{\mathsf{V}_{A}}\mathcal{K}'=0$. 

An explicit expansion of the exponents in~(\ref{abelianShift}) shows that the difference between $\mathcal{K}''$ and~$\mathcal{K}'$ is given by
\begin{align}
    &\mathcal{K}''-\mathcal{K}'= f(\mathbf{Z},\mathbf{S})+\widebar{f}\left(\widebar{\mathbf{Z}}, 
\widebar{\mathbf{S}}\right)+g(\mathbf{Z}, \widebar{\mathbf{S}})+\widebar{g}\left(\widebar{\mathbf{Z}}, 
\mathbf{S}\right)\,,\label{DiffKahler}\\
&\text{where}\quad f(\mathbf{Z},\mathbf{S}) = 
    \frac{e^{D^{(1,0)}}-1}{D^{(1,0)}}\sum_{A=1}^k \mathbf{S}_A f_{\mathsf{W}_A} \,,\quad g(\mathbf{Z}, \widebar{\mathbf{S}}) = \frac{e^{\widebar{D}^{(1,0)}}-1}{\widebar{D}^{(1,0)}}\sum_{A=1}^k \widebar{\mathbf{S}}_A f_{\mathsf{W}_A}\,,\nonumber\\
    & D^{(1,0)} = \sum_{A=1}^k \mathbf{S}_A \nabla_{\mathsf{W}_{A}}^{(1,0)}\,,\quad \widebar{D}^{(1,0)} = \sum_{A=1}^k \widebar{\mathbf{S}}_A \nabla_{\mathsf{W}_{A}}^{(1,0)}\,.\nonumber
\end{align}
In order for the shift (\ref{DiffKahler}) to be a generalized K\"{a}hler transformation~(\ref{genKahtrans1}),  we require the fields $\mathbf{S}_{A}$ to be twisted chiral. 
\end{proof}

Given an invariant K\"ahler potential, one can gauge an Abelian isometry group. However, in the general case with a central extension isometry transformations act non-trivially on the twisted chiral fields (cf. (\ref{NewIsometry})). In this case one can use a general method based on superconnections~\cite{BKK}. Let us illustrate this method with core examples from Section~\ref{HeisenbergAlgebra}. We recall that, in both examples, there are two commuting Killing vector fields, $\mathsf{V}$ and $\mathsf{W}$. The two cases differ by whether their $(1,0)$ parts are linearly independent (Case 1) or linearly dependent (Case 2). 

\subsection{Case 1: $\mathsf{V}^{(1,0)}\wedge \mathsf{W}^{(1,0)}\neq 0$.}
\label{IndependentHolPartsVector} 
The invariant K\"ahler potential is given by:
\begin{align}
    \mathcal{K} = -\frac{ic}{2}\Bigl(\mathbf{W}+\widebar{\mathbf{W}}+\mathbf{S}+\widebar{\mathbf{S}}\Bigr)\Bigl(\mathbf{Z}-\widebar{\mathbf{Z}}\Bigr)+\mathcal{K}_0\Bigl(\mathbf{W}-\widebar{\mathbf{W}},\mathbf{Z}-\widebar{\mathbf{Z}},\cdots\Bigr)\,,\label{InvariantPotentialIndependentCase}
\end{align}
where $\mathbf{S}$ is twisted chiral and isometry transformations acts as $\mathbf{Z}\mapsto \mathbf{Z}+\epsilon_1$, $\mathbf{W}\mapsto \mathbf{W}+\epsilon_2$ and $\mathbf{S}\mapsto \mathbf{S}-\epsilon_2$. In order to gauge this symmetry we will modify the (twisted) chirality conditions for $\mathbf{Z}$, $\mathbf{W}$ and $\mathbf{S}$ by introducing superconnections:
\begin{align}
    &\widebar{\mathcal{D}}_+ \mathbf{Z} + i\widebar{\mathbf{A}}^z_+ = 0\,, && \widebar{\mathcal{D}}_- \mathbf{Z} + i\widebar{\mathbf{A}}^z_- = 0\,, \nonumber\\
    &\widebar{\mathcal{D}}_+ \mathbf{W} + i\widebar{\mathbf{A}}^w_+ = 0\,, && \widebar{\mathcal{D}}_- \mathbf{W} + i\widebar{\mathbf{A}}^w_- = 0\,,\label{ModifiedChiralityConditions}\\
    & \widebar{\mathcal{D}}_+ \mathbf{S} + i\widebar{\mathbf{A}}^s_+ = 0\,, && \mathcal{D}_- \mathbf{S} + i\mathbf{A}^s_- = 0\,.\nonumber
\end{align}
We shall demand $\widebar{\mathbf{A}}^s_+ = -\widebar{\mathbf{A}}^w_+$ to ensure that  we are gauging the same $\epsilon_2$-isometry. One can then show that the solutions to~(\ref{ModifiedChiralityConditions}) are parametrized by standard (twisted) chiral fields and additional gauge superfields: the  real superfields $\mathbf{V}_z,\mathbf{V}_w$   and the semi-chiral superfield~$\mathbf{X}$. The superconnections may be expressed through the latter~\cite{BKK}. Substituting these solutions into~ (\ref{InvariantPotentialIndependentCase}), we arrive at 
\begin{align} \label{invariantpotentialindependentcasegauged}
    \mathcal{K} = -&\frac{ic}{2}\Bigl(\mathbf{W}+\widebar{\mathbf{W}}+\mathbf{S}+\widebar{\mathbf{S}}+i\mathbf{X}-i\widebar{\mathbf{X}}\Bigr)\Bigl(\mathbf{Z}-\widebar{\mathbf{Z}}+i\mathbf{V}_z\Bigr)+\\
    &+\mathcal{K}_0\Bigl(\mathbf{W}+\widebar{\mathbf{W}}+i\mathbf{V}_w,\mathbf{Z}-\widebar{\mathbf{Z}}+i\mathbf{V}_z,\cdots\Bigr),\nonumber
\end{align}
where $\mathbf{Z},\mathbf{W}$ and $\mathbf{S}$ are now the standard (twisted) chiral fields. 

\subsubsection{T-duality.} 
In order to perform T-duality, we shall add terms to the K\"ahler potential ensuring that the gauge fields are pure gauge. In $\mathcal{N}=(2,2)$ superspace these are 
\begin{align}
    \Delta\mathcal{K} = \Bigl(\mathbf{\Sigma}+\widebar{\mathbf{\Sigma}}\Bigr)\mathbf{V}_z + \mathbf{Y}\Bigl(\mathbf{X}-\mathbf{V}_w\Bigr) + \widebar{\mathbf{Y}}\Bigl(\widebar{\mathbf{X}}-\mathbf{V}_w\Bigr)\,,
\end{align}
where $\mathbf{\Sigma}$ is twisted chiral and $\mathbf{Y}$ is right semi-chiral (see \cite{BKK} for details). Now we eliminate $\mathbf{Z},\mathbf{W}$ and $\mathbf{S}$ using gauge transformations and make a shift $\mathbf{X}\mapsto\mathbf{X}+\frac{2i}{c}\mathbf{\Sigma}$. We ultimately arrive at
\begin{align}
    \mathcal{K}+\Delta\mathcal{K} = &\frac{ic}{2}\Bigl(\mathbf{X} - \widebar{\mathbf{X}}\Bigr)\mathbf{V}_z - \Bigl(\mathbf{Y} + \widebar{\mathbf{Y}}\Bigr)\mathbf{V}_w +
    \\
    &+ \mathcal{K}_0\Bigl(i\mathbf{V}_w,i\mathbf{V}_z,\cdots\Bigr) + \mathbf{X}\mathbf{Y} + \widebar{\mathbf{X}}\widebar{\mathbf{Y}}\,.\nonumber
\end{align}

To find the K\"ahler potential of the T-dual, we are left to integrate over $\mathbf{V}_z$ and $\mathbf{V}_w$, i.e. the new K\"{a}hler potential is just a Legendre transform plus some additional terms.

\subsubsection{The quotient.}

If we turn to the quotient construction, we add possible  FI terms to~(\ref{invariantpotentialindependentcasegauged}) and integrate out $\mathbf{V}_z$ and $\mathbf{V}_w$. Effectively, we work with 
\begin{align}
    \mathcal{K} = \frac{ic}{2}\Bigl(\mathbf{X} - \widebar{\mathbf{X}}\Bigr)\mathbf{V}_z + \mathcal{K}_0\Bigl(i\mathbf{V}_w,i\mathbf{V}_z,\cdots\Bigr)+\xi \mathbf{V}_z + \zeta \mathbf{V}_w\,,\label{KahlerQuotient1C}
\end{align}
where $\xi,\zeta \in \mathbb{R}$. One then needs to solve the equations $\dd_{\mathbf{V}_z}\mathcal{K}_0 + \frac{ic}{2}\Bigl(\mathbf{X} - \widebar{\mathbf{X}}\Bigr) + \xi= 0$ and $\dd_{\mathbf{V}_w}\mathcal{K}_0 + \zeta = 0$ and substitute the solutions back into (\ref{KahlerQuotient1C}). 

\subsection{Case 2: $\mathsf{V}^{(1,0)}\wedge \mathsf{W}^{(1,0)}= 0$.}
\label{lindepsec}
According to~(\ref{abelianShift}), the invariant K\"ahler potential takes the  form
\begin{align}
    \mathcal{K} = -\frac{ic}{\uplambda - \smallthickbar{\uplambda}}\left[\left(\frac{\mathbf{Z}-\widebar{\mathbf{Z}}}{2}+\mathbf{S}\frac{\uplambda - \widebar{\uplambda}}{2}\right)^2+\left(\frac{\mathbf{Z}-\widebar{\mathbf{Z}}}{2}+\widebar{\mathbf{S}}\frac{\uplambda - \widebar{\uplambda}}{2}\right)^2\right]+\Tilde{\mathcal{K}}_0\,,
\end{align}
where $\uplambda$ is a holomorphic function of any remaining coordinates, and the isometry transformations now act as $\mathbf{Z}\mapsto \mathbf{Z} + \epsilon_1 + \uplambda\epsilon_2$ and $\mathbf{S} \mapsto \mathbf{S} - \epsilon_2$. To perform the gauging, one uses a modification of the (twisted) chirality conditions as in Section \ref{IndependentHolPartsVector}, arriving at
\begin{align} \label{KahpotgaugedCase2}
    \mathcal{K} = -\frac{ic}{\uplambda - \smallthickbar{\uplambda}}&\Big[\left(\frac{\mathbf{Z}-\widebar{\mathbf{Z}}+i\mathbf{V}}{2}+\left(\mathbf{S}-\mathbf{X}\right)\frac{\uplambda - \widebar{\uplambda}}{2}\right)^2+\\
    &+\left(\frac{\mathbf{Z}-\widebar{\mathbf{Z}}+i\widebar{\mathbf{V}}}{2}+\left(\widebar{\mathbf{S}}-\widebar{\mathbf{X}}\right)\frac{\uplambda - \widebar{\uplambda}}{2}\right)^2\Big]+\mathcal{K}_0, \nonumber
\end{align}
where $\mathbf{V}\equiv \mathbf{V}_1 + \widebar{\uplambda}\mathbf{V}_2$ is a complex superfield and $\mathbf{X}$ is left semi-chiral. 

\subsubsection{T-duality.}

In order to perform T-duality, one introduces an additional twisted chiral field $\mathbf{\Sigma}$ and a right semi-chiral field $\mathbf{Y}$ and adds the term
\begin{align}
\Delta\mathcal{K} = (\mathbf{\Sigma}+\widebar{\mathbf{\Sigma}})\mathbf{V}_1 + \mathbf{Y}(\mathbf{V}_2-i\mathbf{X}) + \widebar{\mathbf{Y}}(\mathbf{V}_2+i\widebar{\mathbf{X}})  
\end{align}
to ensure that the $\mathbf{V}_i$'s and $\mathbf{X}$ are pure gauge. One then eliminates the $\mathbf{V}_i$'s and shifts~$\mathbf{X}$, which yields
\begin{align}
    \mathcal{K} = \frac{2i}{c}\frac{\left(\mathbf{Y}+\widebar{\mathbf{Y}}\right)^2}{\uplambda-\smallthickbar{\uplambda}} + \mathcal{K}_0- i\mathbf{Y}\widebar{\mathbf{X}} +i \widebar{\mathbf{Y}}\mathbf{X} \,.
\end{align}
In fact, the last two terms imply that $\widebar{\mathbf{Y}} =\mathbf{\Phi}$ is chiral. Thus, the resulting geometry in this case is K\"ahler.  

\subsubsection{The quotient.}

If one wishes to perform the quotient instead, one adds possible FI-terms to~(\ref{KahpotgaugedCase2}), i.e. $\xi \mathbf{V}_1 + \zeta \mathbf{V}_2$ for $\xi,\zeta \in \mathbb{R}$, and eliminates the  $\mathbf{V}_i$'s. As a result, one finds that the K\"ahler potential of the quotient is given by
\begin{align}
    \mathcal{K} = \mathcal{K}_0 + \frac{i}{c}\frac{1}{\uplambda-\smallthickbar{\uplambda}}\left[\left(\uplambda\xi-\zeta\right)^2 + \left(\smallthickbar{\uplambda}\xi - \zeta\right)^2\right]\,.
\end{align}

\subsection{General case.} To describe the gauging in the general case, one should make use of Proposition~\ref{uniformprop}. It states that any K\"ahler geometry whose Abelian isometry algebra is centrally extended may be obtained by a usual K\"ahler quotient of a universal geometry termed~$\mathcal{N}$. The non-invariant part of $\mathcal{N}$ that sources the central extension is simply a sum of several pieces, each  of the type appearing in~(\ref{InvariantPotentialIndependentCase}). The gauging of each individual piece is described by~(\ref{invariantpotentialindependentcasegauged}), and the full gauged potential is obtained by superimposing these pieces.

We note that one could also follow this approach in the example of Case 2. In Appendix~\ref{appKahquot}, we elaborate on how this geometry may be obtained as a K\"ahler quotient of the geometry of type $\mathcal{N}$, which in this setting is simply the geometry~(\ref{InvariantPotentialIndependentCase}) of Case 1. The alternative derivation was provided  in Section~\ref{lindepsec} to present a different viewpoint. The end result, of course, remains the same.

\section{Conclusion} 

In the present paper we studied  K\"ahler geometries with an Abelian algebra of isometries, such that the corresponding algebra of moment maps has a central extension. As we explained, this is tantamount to the fact that the K\"ahler potential is not truly invariant under all of the symmetry transformations, but is only invariant up to K\"ahler transformations. It has been known for a long time that in this case the naive gauging of such isometries is obstructed, which is why they may be called \emph{anomalous}. Nevertheless, we have shown that such isometries may still be gauged, albeit in the realm of generalized K\"ahler geometry: one first introduces auxiliary twisted chiral fields to make the potential invariant and then utilizes an extended gauge multiplet $\left(\mathbf{V}, \mathbf{X}\right)$ consisting of a real and a semi-chiral superfield. This construction is relevant for performing generalized K\"ahler quotients or T-dualities of the original K\"ahler geometries: in both cases one typically arrives at a generalized K\"ahler geometry, for example, a sigma model interacting with a $\beta\gamma$-system. Future directions include generalizing to non-Abelian isometry algebras and possibly to the infinite-dimensional geometries admitting central extensions (for example, the Virasoro and loop groups).

\vspace{1cm}
\textbf{Acknowledgments.} We acknowledge support from the Russian Science Foundation under grant \href{https://rscf.ru/project/7itJjWqd9pYqz47liTvCgk4qUmdtral1n3ecf41C_0I7jWkQGL8qIR8wWeq93dHAgwWeRIZfZPc~/}{№ 25-72-10177}. The work of A.K. was supported by the Foundation for the Advancement of Theoretical Physics and Mathematics ``BASIS". We would like to thank  M.~Alfimov, A.~Belavin, S.~Fedoruk, M.~Gritskov, E.~Ivanov, V.~Krivorol, U.~Lindstr\"om, A.~Litvinov, M.~Ro\v{c}ek, A.~Smilga, M.~Vasiliev for useful discussions and  S.~Kutsubin for a collaboration on related topics. 

\appendix

\section{K\"ahler quotient in the Heisenberg case}\label{appKahquot}

Here we will show how the K\"ahler potential~(\ref{abelianKah2}) arises from~(\ref{abelianKah1}) via a K\"ahler quotient. We thus start from the K\"ahler potential~(\ref{abelianKah1}),
\begin{align}\label{invpotapp}
    \mathcal{K}=-{ic\over 2}(\mathrm{w}+\smallthickbar{\mathrm{w}})(\mathrm{z}-\smallthickbar{\mathrm{z}})+\mathcal{K}_0(\mathrm{w}-\smallthickbar{\mathrm{w}}, \mathrm{z}-\smallthickbar{\mathrm{z}}, \cdots)+a(\mathrm{r})\,{\mathrm{w}^2\over 2}+\widebar{a}(\smallthickbar{\mathrm{r}}) {\smallthickbar{\mathrm{w}}^2\over 2}\,,
\end{align}
where we have performed an additional K\"ahler transformation  as in~(\ref{anomalousKgen}) (its  meaning will also become clear shortly). 

The corresponding Killing vector fields are $\mathsf{V}^{(1,0)}=\dd_\mathrm{z}$ and ${\mathsf{W}^{(1,0)}=\dd_\mathrm{w}}$. We assume that there is an additional isometry generated by the vector field
\begin{align}
    \mathsf{U}^{(1,0)}=\mathrm{A}(\mathrm{r}){\dd \over \dd \mathrm{w}}+\mathrm{B}(\mathrm{r}){\dd \over \dd \mathrm{z}}
\end{align}
The coefficients are independent of $\mathrm{w}$ and $\mathrm{z}$, since $[\mathsf{V}, \mathsf{U}]=[\mathsf{W}, \mathsf{U}]=0$.

By construction, the vector fields $\mathsf{V}, \mathsf{W}$ lead to a non-trivial central extension: ${\dutchcal{C}(\mathsf{V}, \mathsf{W})=c}$. In contrast, the vector field $\mathsf{U}$ is subject to the condition ${\dutchcal{C}(\mathsf{V}, \mathsf{U})=\dutchcal{C}(\mathsf{W}, \mathsf{U})=0}$, i.e. $\mathsf{U}\in \mathrm{Ker}\,\dutchcal{C}$. By Proposition~\ref{abeliantrivcocycleprop} this implies that $\mathrm{A}(\mathrm{r})\nequiv 0$ and $\mathrm{B}(\mathrm{r})\nequiv 0$. Moreover, one can make the potential invariant under, say, $\mathsf{V}$ and $\mathsf{U}$ simultaneously. 

Any two potentials of the form~(\ref{invpotapp}) with the required invariance properties differ by a function $\widehat{\mathcal{K}}$ invariant under all three vector fields $\mathsf{V}, \mathsf{W}, \mathsf{U}$. We will thus look for $\mathcal{K}_0$ in the form
\begin{align}
    &\mathcal{K}_0=\mathrm{C}(\mathrm{r}, \smallthickbar{\mathrm{r}})\,(\mathrm{w}-\smallthickbar{\mathrm{w}})(\mathrm{z}-\smallthickbar{\mathrm{z}})+{\mathrm{D}(\mathrm{r}, \smallthickbar{\mathrm{r}})\over 2}\,(\mathrm{z}-\smallthickbar{\mathrm{z}})^2+\widehat{\mathcal{K}}(x)\,,\\
    &\textrm{where}\quad\quad x=(\mathrm{A}-\widebar{\mathrm{A}})(\mathrm{z}-\smallthickbar{\mathrm{z}})-(\mathrm{B}-\widebar{\mathrm{B}})(\mathrm{w}-\smallthickbar{\mathrm{w}})\,.
\end{align}
It suffices to show that the functions $\mathrm{C}$, $\mathrm{D}$ and $a(\mathrm{r})$ may be chosen in such a way that the potential $\mathcal{K}$ constructed from $\mathcal{K}_0$ using formula~(\ref{invpotapp}) is invariant under $\mathsf{U}$. A direct calculation gives the values
\begin{align}
    \mathrm{C}={ic\over 2}\,\frac{\mathrm{B}+\widebar{\mathrm{B}}}{\mathrm{B}-\widebar{\mathrm{B}}}\,,\quad\quad \mathrm{D}=i\,c\,\frac{\mathrm{AB}-\widebar{\mathrm{A}}\widebar{\mathrm{B}}}{(\mathrm{B}-\widebar{\mathrm{B}})^2}\,,\quad\quad a=i\,c\,{\mathrm{B}+\kappa\over \mathrm{A}}\,,
\end{align}
with $\kappa$ a real constant. It is then easy to see that $c(\mathsf{W}, \mathsf{U})=0$ implies $\kappa=0$.

To simplify the expressions, it will be useful to rectify the vector field $\mathsf{U}^{(1,0)}=\dd_{\mathrm{u}}$ by making the change of variables $\mathrm{w}=\mathrm{A} \,\mathrm{u}, \mathrm{v}=\mathrm{v}'+\mathrm{B}\,\mathrm{u}$ (we will then drop the prime, so that we effectively substitute $\mathrm{v}\mapsto \mathrm{v}+\mathrm{B}\,\mathrm{u}$). In these variables, the potential takes the form
\begin{align}
    \mathcal{K}={i\, c\over 2}\,\frac{\mathrm{AB} (\widebar{\mathrm{B}}(\mathrm{u}-\smallthickbar{\mathrm{u}})+\mathrm{z}-\smallthickbar{\mathrm{z}})^2-\widebar{\mathrm{A}}\widebar{\mathrm{B}}(\mathrm{B}(\mathrm{u}-\smallthickbar{\mathrm{u}})+\mathrm{z}-\smallthickbar{\mathrm{z}})^2}{(\mathrm{B}-\widebar{\mathrm{B}})^2}+&\\  \nonumber +\,\widehat{\mathcal{K}}\left(\mathrm{u}-\smallthickbar{\mathrm{u}}-{\mathrm{A}-\widebar{\mathrm{A}}\over \widebar{\mathrm{A}}\mathrm{B}-\widebar{\mathrm{B}}\mathrm{A}}\,(\mathrm{z}-\smallthickbar{\mathrm{z}})\right)
\end{align}

As usual, in order to gauge the shift symmetry in $\mathrm{u}$ generated by $\mathsf{U}$ we replace $\mathrm{u}-\smallthickbar{\mathrm{u}}\mapsto \mathrm{u}-\smallthickbar{\mathrm{u}}+i\,\mathbf{V}$, where $\mathbf{V}$ is the real gauge superfield, and choose the gauge $\mathrm{u}=\smallthickbar{\mathrm{u}}=0$. We may also add an FI-term of the form $\xi \,\mathbf{V}$, with $\xi\in \mathbb{R}$. Subsequently performing the redefinition 
\begin{align}
    \mathbf{V}\mapsto \mathbf{V}-i{\mathrm{A}-\widebar{\mathrm{A}}\over \widebar{\mathrm{A}}\mathrm{B}-\widebar{\mathrm{B}}\mathrm{A}}\,(\mathrm{z}-\smallthickbar{\mathrm{z}})\,,
\end{align}
we arrive at the final answer for the gauged potential:
\begin{align}\nonumber
    &\mathcal{K}_g={i\,c\over 2}\,\frac{\mathrm{B}\widebar{\mathrm{A}}-\mathrm{A}\widebar{\mathrm{B}}}{(\mathrm{B}-\widebar{\mathrm{B}})^2}|\mathrm{B}|^2\,\mathbf{V}^2+\widehat{\mathcal{K}}(i\mathbf{V})-{i\,c\over 2}\,\frac{(\mathrm{z}-\smallthickbar{\mathrm{z}})^2}{\uplambda-\widebar{\uplambda}}+{f_0+\widebar{f}_0\over \uplambda-\widebar{\uplambda}}\,(\mathrm{z}-\smallthickbar{\mathrm{z}})\,,\quad\quad\\& \textrm{where}\quad\quad \uplambda=-{\mathrm{B}\over \mathrm{A}}\,,\quad f_0={i\,\xi\over \widebar{\mathrm{A}}}\,.
\end{align}
One observes that the $\mathrm{z}$-dependent part has fully decoupled. Assuming that the equation ${\dd_\mathbf{V} \mathcal{K}_g}=0$ has a unique solution\footnote{This is a rather mild assumption. Indeed, ${\dd^2 \mathcal{K}_g\over \dd \mathbf{V}^2}>0$ due to the fact that the metric is positive-definite,  ${\dd^2 \mathcal{K}\over \dd \mathrm{u} \dd \smallthickbar{\mathrm{u}}}>0$. Thus, if a solution exists, it is unique. A sufficient condition for the existence of a solution would be a uniform bound ${\dd^2 \mathcal{K}\over \dd \mathrm{u} \dd \smallthickbar{\mathrm{u}}}>\epsilon>0$.} $\mathbf{V}=\mathbf{V}(\mathrm{r}, \smallthickbar{\mathrm{r}})$, one can eliminate $\mathbf{V}$ to arrive at the following K\"ahler potential on the quotient:
\begin{align}
    \mathcal{K}_g=-{i\,c\over 2}\,\frac{(\mathrm{z}-\smallthickbar{\mathrm{z}})^2}{\uplambda-\widebar{\uplambda}}+{f_0+\widebar{f_0}\over \uplambda-\widebar{\uplambda}}\,(\mathrm{z}-\smallthickbar{\mathrm{z}})+\mathrm{F}(\mathrm{r}, \smallthickbar{\mathrm{r}})\,,
\end{align}
which is exactly expression~(\ref{Kahpottoruslinterm}). As observed there, one can effectively get rid of the linear term in $\mathrm{z}-\smallthickbar{\mathrm{z}}$ by a shift $\mathrm{z}\mapsto \mathrm{z}+{1\over i c} f_0$,  arriving at~(\ref{abelianKah2}).

It is also instructive to note that, upon setting $\mathsf{U}^{(1,0)}\equiv 0$ on the quotient, one finds $\mathsf{V}^{(1,0)}={\dd_\mathrm{z}}$ and $\mathsf{W}^{(1,0)}=\uplambda(\mathrm{r})\,{\dd_\mathrm{z}}$,  which brings us back to the setup of Section~\ref{HeisenbergAlgebra}.

\vspace{0.3cm}    
    \setstretch{0.8}
    \setlength\bibitemsep{5pt}
    \printbibliography
    
\end{document}